\documentclass[a4paper,UKenglish,cleveref, autoref, thm-restate,numberwithinsect]{lipics-v2021}

\pdfoutput=1 %
\hideLIPIcs  %

\graphicspath{{graphics/}}%

\title{Point Set Embeddability with List Constraints}

\author{Thomas Depian}{TU Wien, Austria}{tdepian@ac.tuwien.ac.at}{https://orcid.org/0009-0003-7498-6271}{Supported by the Vienna Science and Technology Fund (WWTF) [10.47379/ICT22029].}

\author{Joseph Dorfer}{TU Graz, Austria}{joseph.dorfer@tugraz.at}{https://orcid.org/0009-0004-9276-7870}{Supported by the Austrian Science Fund (FWF) 10.55776/DOC183.}

\author{Boris Klemz}{Universität Würzburg, Germany}{firstname.lastname@uni-wuerzburg.de}{https://orcid.org/0000-0002-4532-3765}{}

\author{Matthias Pfretzschner}{University of Passau, Germany}{pfretzschner@fim.uni-passau.de}{https://orcid.org/0000-0002-5378-1694}{}

\author{Lena Schlipf}{Universit\"at T\"ubingen, Germany}{lena.schlipf@uni-tuebingen.de}{https://orcid.org/0000-0001-7043-1867}{Supported by the  German Research Foundation (DFG grant SCHL 2331/1-1)}

\authorrunning{T. Depian, J. Dorfer, B. Klemz, M. Pfretzschner, and L. Schlipf} %

\Copyright{Thomas Depian, Joseph Dorfer, Boris Klemz, Matthias Pfretzschner, and Lena Schlipf} %



\acknowledgements{This
work started at the 20th European Research Week on Geometric Graphs (GGWeek'25)
in Chorin, Germany. We thank
the organizers for fruitful and nice atmosphere.}%

\nolinenumbers %

\EventEditors{Maarten L\"{o}ffler and Silvia Miksch}
\EventNoEds{2}
\EventLongTitle{34th International Symposium on Graph Drawing and Network Visualization (GD 2026)}
\EventShortTitle{GD 2026}
\EventAcronym{GD}
\EventYear{2026}
\EventDate{August 19--21, 2026}
\EventLocation{Ontario, Canada}
\EventLogo{}
\SeriesVolume{396}
\ArticleNo{33}

\usepackage{xspace}
\usepackage{todonotes}
\usepackage{complexity}
\usepackage{framed}
\usepackage{mathtools}
\usepackage{booktabs}
\usepackage{amsmath}
\usepackage{pifont}
\usepackage{cite}

\usetikzlibrary{arrows.meta,patterns,ipe}

\newboolean{long}
\newcommand{\AppendixSymbol}{\ding{72}}
\setboolean{long}{true} %
\ifthenelse{\boolean{long}}{
	\NewDocumentEnvironment{prooflater}{m}{\begin{proof}}{\end{proof}\ignorespacesafterend}
	\NewDocumentEnvironment{proofsketch}{o +b}{}{\ignorespacesafterend}
	\newcommand{\restateref}[1]{}
	\NewDocumentEnvironment{statelater}{m}{}{}

	\NewDocumentCommand{\onlyShort}{+m}{}
	\NewDocumentCommand{\onlyLong}{+m}{#1}
	\NewDocumentCommand{\shortLong}{+m +m}{#2}
}{
	\NewDocumentEnvironment{prooflater}{m +b}{%
		\expandafter\global\expandafter\def\csname#1\endcsname{\begin{proof}#2\end{proof}}%
	}{\ignorespacesafterend}
	\NewDocumentEnvironment{proofsketch}{O{Proof sketch.}}{\begin{proof}[#1]}{\end{proof}\ignorespacesafterend}
	\usepackage{apptools}

	\newcommand{\restateref}[1]{[\IfAppendix{\AppendixSymbol{}}{\AppendixSymbol{}}]}
	
	\NewDocumentEnvironment{statelater}{m +b}{%
		\expandafter\global\expandafter\def\csname#1\endcsname{#2}%
	}{\ignorespacesafterend}

	\NewDocumentCommand{\onlyShort}{+m}{#1}
	\NewDocumentCommand{\onlyLong}{+m}{}
	\NewDocumentCommand{\shortLong}{+m +m}{#1}
}
\let\oldrestatable\restatable
\def\restatable{\expandafter\oldrestatable}

\newcommand{\probname}[1]{{\normalfont\textsc{#1}}}

\newcommand{\ListPointEmbedding}{\probname{GLPSE}\xspace}
\newcommand{\PointEmbedding}{\probname{GPSE}\xspace}

\newcommand{\ConvexListPointEmbedding}{\probname{C-GLPSE}\xspace}

\newcommand{\ListPointEmbeddingLong}{\probname{Geometric List Point Set Embeddability}\xspace}
\newcommand{\PointEmbeddingLong}{\probname{Geometric Point Set Embeddability}\xspace}

\newcommand{\PartialListPointEmbedding}{\probname{P-\ListPointEmbedding}\xspace}

\newcommand{\Instance}{\ensuremath{\mathcal{I}}\xspace}
\newcommand{\InstanceLong}{\ensuremath{(G, S, L)}\xspace}

\definecolor{colorAi}{rgb}{0 0.439 0.38} %
\definecolor{colorAii}{rgb}{0.263 0.42 0.733} %
\newcommand{\Ai}{\probname{ALG1}\xspace}
\newcommand{\Aii}{\probname{ALG2}\xspace} %
\newcommand{\algs}[2]{{{[} \textcolor{colorAi}{#1} {|} \textcolor{colorAii}{#2}
{]}}}

\NewDocumentCommand{\Drawing}{o}{\ensuremath{\Gamma\IfNoValueF{#1}{{(#1)}}}\xspace}

\newcommand{\vcn}{\textsf{\textup{vcn}}}

\newcommand{\nadd}{\ensuremath{n_{\mathrm{add}}}\xspace}
\newcommand{\Vadd}{\ensuremath{V_{\mathrm{add}}}\xspace}

\newcommand{\Size}[1]{\ensuremath{\left\vert #1 \right\vert}}
\newcommand{\BigO}[1]{\ensuremath{\mathcal{O}(#1)}}
\DeclareMathOperator*{\partition}{\mathbin{\mathaccent\cdot\cup}}

\theoremstyle{claimstyle}
\newtheorem{ob}[claim]{Observation}

\newcommand{\NewText}[1]{{#1}}

\begin{document}

\maketitle

\begin{abstract}%
Deciding whether a given graph admits a planar straight-line drawing where each vertex is placed on some point from a given finite point set is known as \probname{Point Set Embeddability} and \NewText{is} a classical problem in graph drawing.
In this paper, we study the more general embeddability question where
the placement of each vertex $v$ is restricted to a list $L(v)$ of admissible points.

We first study the case where the given point set is in convex position.
We show that this case is \NP-hard even if the given graph is a matching and bi-labeled, i.e., each vertex has at most 2 admissible points.
On the positive side, we present
two efficient algorithms for the case where the given graph~$G$ is connected (and not necessarily bi-labeled):
if $G$ is equipped with a combinatorial embedding that needs to be respected, we can solve the problem in polynomial time;
otherwise we can solve it in \FPT-time with regard to the maximum vertex degree.
In particular, this answers an open question by Frati, Glisse, Lenhart, Liotta, Mchedlidze, and Nishat~[GD'13]. %

We then turn our attention to the more general case where the given point set is not necessarily in convex position. Here, we show \NP-hardness for bi-labeled paths; notably these graphs have a %
unique combinatorial embedding and maximum degree two, thereby creating a stark contrast to our algorithms for the convex case. %
We also present an \FPT-algorithm with respect to the vertex cover number for the special case of bi-labeled graphs.
We complement this latter result by establishing para\NP-hardness in the tri-labeled setting for vertex cover number~2 and polynomial-time solvability for vertex cover number 1 and arbitrary $L$.

Finally, we study optimization and extension variants, where we want to maximize the number of edges or extend a partial drawing, respectively.
For the former, we show \APX-hardness and for the latter, we provide a parameterized complexity dichotomy under natural extension parameters.

\end{abstract}

\section{Introduction}
\label{sec:introduction}
Given an $n$-vertex graph $G$, a finite point set $S$ in the plane, and a set of \emph{labels} $L(v) \subseteq S$ for each vertex $v$ of $G$, the problem \ListPointEmbeddingLong (\ListPointEmbedding) asks whether $G$ admits a planar straight-line drawing where each vertex $v$ is placed on one of its labels.
This problem can be viewed as a generalization or variation of several well-studied problems in graph drawing and computational geometry.
In particular, the well-known problem \PointEmbeddingLong (\PointEmbedding) is the special case of \ListPointEmbedding where each vertex $v$ may be placed on any point of $S$, i.e., $L(v) = S$.
While it is known that, in this setting, outerplanar graphs can be embedded on any point set $S$ in general position with $|S| \geq n$ \cite{CU.SLE.1996, GMP+.Ept.1991}, \PointEmbedding becomes \NP-hard for non-outerplanar input graphs~\cite{Cab.Pev.2006}.
On the positive side, if the combinatorial embedding is part of the input and both the treewidth and the face degree of the graph is constant, then \PointEmbedding can be solved in polynomial time~\cite{BV.pse.2012}.
However, both conditions are necessary: If one of them is dropped, the problem is \NP-hard, even when the combinatorial embedding is part of the input~\cite{BV.pse.2012,DM.HPS.2012}.
Other research on \PointEmbedding also focused on embeddings with few bends per edge (see, e.g.,~\cite{KW.EVP.1999,PW.EPG.2001}) and on finding universal point sets, i.e., point sets that allow for a planar straight-line embedding of any planar graph on~$n$ vertices (see, e.g.,~\cite{ABB+.SUP.2018,BCD+.SUP.2014,CHK.UPS.2015,DPP.Hdp.1990,Kur.Alb.2004,SSS.NUP.2020,Sch.EPG.1990}).
We remark, however, that positive results and universal point sets for \PointEmbedding do not necessarily translate to \ListPointEmbedding due to the constraints on the allowed placement for each vertex.

Previous work also focused on a special case of \ListPointEmbedding where the point set $S$ is partitioned into $k$ color classes and each vertex is assigned to a specific color class.
While research mostly focused on finding drawings with few bends per edge (see, e.g.,~\cite{BDGL.Dcg.2008,DGDL+.kCP.2007,GGL+.CCC.2020,GJL.CPS.2021,DGLT.EGT.2006,DGLT.DCG.2008}), more results for straight-line embeddings are known for the case $k = 2$.
In particular, Frati, Glisse, Lenhart, Liotta, Mchedlidze, and Nishat~\cite{FGL+.PSE.2013} showed that deciding whether a bi-colored tree admits a planar straight-line embedding is \NP-complete, even if the points are in convex position.
They leave the complexity question for binary trees open, stating that they suspect it is \NP-hard as well.
However, if the colors of the convex points alternate, the embeddability question can be answered in linear time~\cite{FGL+.PSE.2013}.

Another closely related problem is \probname{Geometric Partial Drawing Extensibility}, which asks whether a given straight-line drawing of a subgraph $H \subseteq G$ can be extended to a straight-line drawing of $G$ without modifying the drawing of~$H$.
While the problem is linear-time solvable for topological drawings~\cite{ADBF+.TPP.2015}, it was shown to be \NP-hard for straight-line drawings~\cite{Pat.EPS.2006}, and it remains open whether it is $\exists\mathbb{R}$-complete~\cite{SCM.ETR.2024}.
However, polynomial-time algorithms exist for special cases and other variants (see, e.g.,~\cite{CET+.DGP.2012,MNR.ECP.2015,PW.EPG.2001}).
If we restrict the possible positions of the missing vertices to a finite set $S'$, this problem can also be modeled as a variant of \ListPointEmbedding where $|L(u)| = 1$ for all $u \in V(H)$ and $L(v) = S'$ for all $v \in V(G)\setminus V(H)$.

\begin{figure}[t]
	\centering
	\input{graphics/overview}
	\caption{Complexity landscape of \ListPointEmbedding. The colors indicate polynomial-time solvable (dark blue), \FPT\ (light blue), \XP-tractable and \W[1]-hard (orange), and \NP-hard (red) settings. Positive results hold for arbitrary $L$, whereas negative results in the non-extension setting hold for bi-labeled instances except for results marked by $\dag$ (algorithms for bi-labeled instances) and $\ddag$ (hardness for tri-labeled instances).
    The only result that is not depicted is our \APX-hardness (\Cref{thm:preapx}).}
	\label{fig:overview}
\end{figure}

\subparagraph{Contribution.}
\enlargethispage{1\baselineskip}
We give an extensive characterization of the (parameterized) complexity of \ListPointEmbedding, which generalizes all aforementioned problems.  
In \Cref{sec:convex}, we 
consider the restricted case where the point set $S$ lies in convex position.
Here, we show that \ListPointEmbedding remains \NP-complete %
even if every vertex $v$ is bi-labeled (i.e., $|L(v)| \leq 2$) and the graph is a matching or a tree of height 2.
On the positive side, for connected input graphs with arbitrary label sizes, we give a polynomial-time algorithm for instances equipped with a combinatorial embedding and we show that the problem is \FPT\footnote{We assume familiarity with basic concepts of \emph{parameterized complexity theory}, in particular, the notions of \emph{fixed-parameter tractability} (\FPT), \XP, and \W[1]-hardness; see~\cite{CFK+.PA.2015} for an introduction.} with respect to the maximum vertex degree,
disproving the aforementioned suspicion of Frati et al.~\cite{FGL+.PSE.2013} that the embeddability problem for binary trees on 2-colored  convex point sets remains \NP-complete (unless $\P=\NP$).

In \Cref{sec:general}, we turn to general point sets.
In contrast to the convex case, \ListPointEmbedding here becomes \NP-complete, even if the graph is a bi-labeled cycle or path (and thus has a unique combinatorial embedding and maximum degree $2$).
In combination with the hardness for bi-labeled matchings, this already rules out fixed-parameter-tractability for almost all structural graph parameters.
To complement this, we show that, for bi-labeled graphs, \ListPointEmbedding is \FPT\ with respect to the vertex cover number.
This result is tight, because the problem becomes \NP-complete if the input graph consists of two tri-labeled (i.e., $\Size{L(v)} \leq 3$ for every $v \in V(G)$) stars and thus has vertex cover number 2.
In contrast, we establish polynomial-time solvability for vertex cover number 1 and arbitrary~$L$.
Moreover, we extend these hardness results to the optimization problem of computing a realization of a subgraph with maximum number of edges.
In particular, we establish \APX-hardness even for bi-labeled graphs with vertex cover number~2, implying that no polynomial-time approximation scheme (PTAS) exists unless $\P = \NP$.

Finally, in \Cref{sec:extension}, we focus on additional parameters in light of the partial drawing extension problem.
While \ListPointEmbedding is trivially in \XP\ parameterized by the number of vertices with more than one label, we show that the problem is \W[1]-hard for the same parameter.
However, we show that if we combine this parameter with the \emph{surplus} $|S| - n$, the problem becomes \FPT.
We summarize our results in \Cref{fig:overview}.

\NewText{Note that we allow, in contrast to some earlier work~\cite{Cab.Pev.2006}, more points in $S$ than vertices in $G$, i.e., $\Size{S} > \Size{V(G)}$.
While all of our positive results hold in this more general variant, also some of the negative results can be readily extended to the more restrictive variant with $\Size{S} = \Size{V(G)}$ by introducing suitable dummy vertices.
Finally, observe that in our last section, we also quantify this ``additional generality'' by considering the surplus as parameter.
}

\onlyShort{
\smallskip
\noindent
\emph{Details for statements with a \AppendixSymbol{} %
\NewText{can be found in the full version~\cite{ARXIV}.}
}
}

\section{Preliminaries}
\label{sec:preliminaries}
For an integer $q \geq 1$, we use $[q]$ as a shorthand for the set $\{1, 2, \ldots, q\}$.
We assume familiarity with standard graph terminology~\cite{Die.GT4.2012}.
Let $G = (V, E)$ be a a simple and undirected graph with vertex set $V(G)$ and edge set $E(G)$.
A \emph{straight-line drawing} \Drawing of $G$ maps each vertex $v \in V(G)$ to a point $\Drawing(v) \subseteq \mathbb{R}^2$ and each edge $\{u,v\}\in E(G)$ to the straight-line segment $\Drawing(u)\Drawing(v)$.
A drawing is \emph{planar} if every pair of edges intersect only at a common endpoint.

Let $S \subset \mathbb{R}^2$ be a point set in the plane with $\Size{S} \geq \Size{V(G)}$ and let $L\colon V(G) \to 2^S$ be a function that specifies for each vertex $v\in V(G)$ its set $L(v) \subseteq S$ of \emph{admissible} points.
A \emph{$L$-realization} $\Drawing\colon V \to S$ of $G$ on $S$ is a planar straight-line drawing of $G$ where every vertex $v \in V(G)$ is placed on an admissible point $\Drawing(v) \in L(v)$.
We use \emph{realization} as a synonym for $L$-realization if the function $L$ is clear from the context.
In this paper, we study the problem \ListPointEmbedding, which asks whether there exists a $L$-realization of $G$ on $S$.

For an instance $\Instance = \InstanceLong$ of \ListPointEmbedding, we define $n \coloneqq \Size{V(G)}$, $m \coloneqq \Size{E(G)}$, $s \coloneqq \Size{S}$, $\ell \coloneqq \max_{v \in V(G)} \Size{L(v)}$, and $\Delta \coloneqq \max_{v \in V(G)} \deg(v)$.
Without loss of generality, we assume $G$ to be planar and have $m \in \BigO{n}$.
We call $G$ bi- or tri-labeled if $\ell \leq 2$ or $\ell \leq 3$, respectively.

\section{Embeddings on Convex Point Sets}
\label{sec:convex}
We start our investigation of \ListPointEmbedding by focusing on the case where $S$ is in convex position.
We refer to this restricted case as \ConvexListPointEmbedding and assume that the points of $S$ are given in their counterclockwise order on the convex hull.

\subsection{Embedding Bi-Labeled Matchings and Trees is Hard}
\label{sec:convex-matching}

\begin{figure}
        \centering
        \includegraphics[page=4,width=\linewidth]{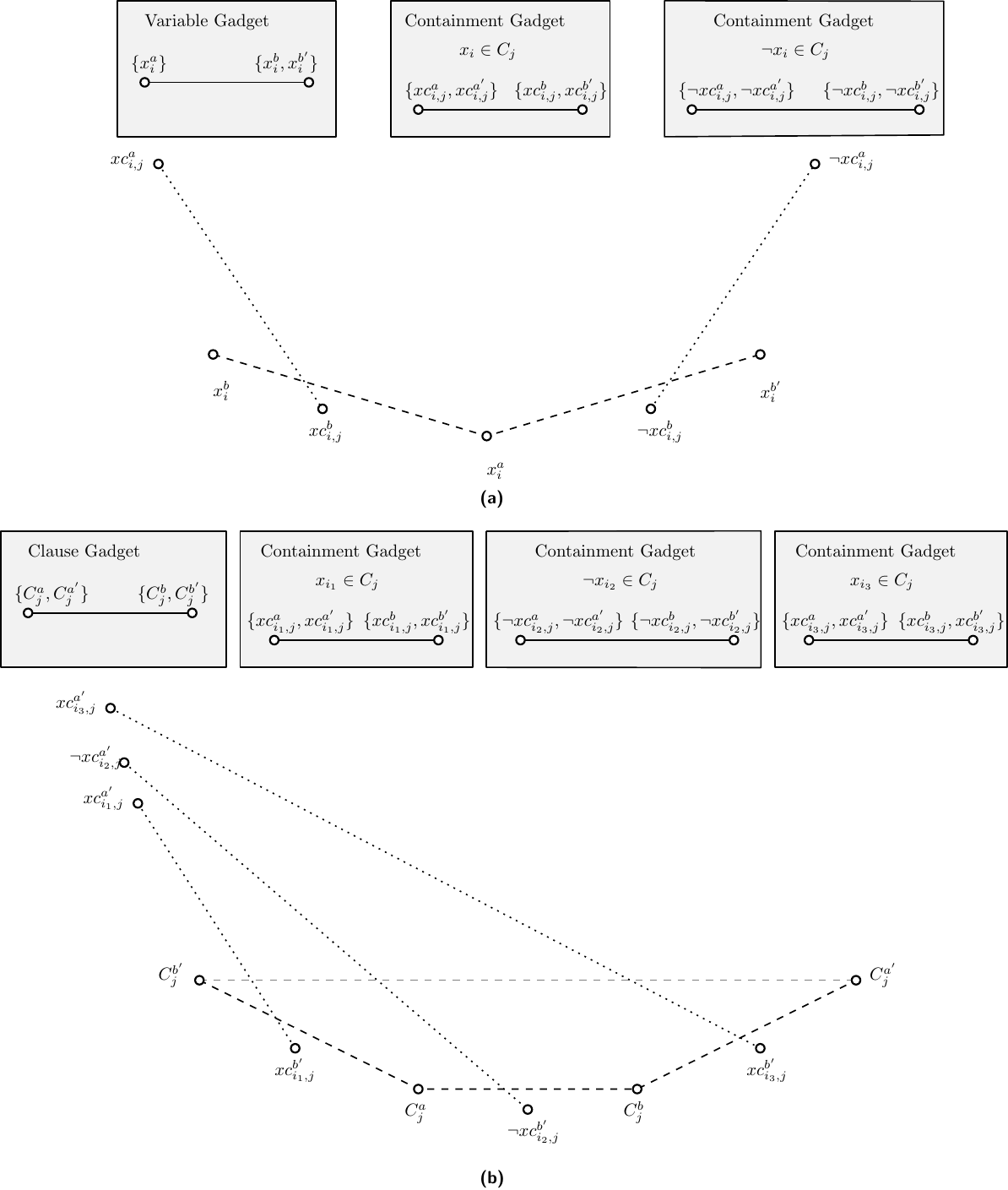}
        \caption{\textbf{\textsf{(a)}} Variable gadget with the first part of the containment gadget; and \textbf{\textsf{(b)}} clause gadget with the second part of the containment gadget.}
        \label{fig:hardmat}
\end{figure}
{
We show that \ConvexListPointEmbedding is \NP-complete even for bi-labeled graphs which are matchings or trees of height two by reducing from \probname{3SAT}. We can assume that clauses contain exactly three literals, as our construction permits copies of the same literal in a clause.

\begin{restatable}\restateref{thm:hardness-matching}{theorem}{theoremHardnessMatching}
	\label{thm:hardness-matching}
    \ConvexListPointEmbedding is \NP-complete even for bi-labeled matchings or bi-labeled trees of height $2$. 
\end{restatable}
\begin{proofsketch} %
We first construct an instance $\Instance = \InstanceLong$ of \ConvexListPointEmbedding where $G$ is a matching as follows (we discuss the modification for trees in the full proof\NewText{~\cite{ARXIV}}).
\NewText{
We will represent variables, clauses, and variable-clause containment relations via edges in $G$.
In particular, all variable edges have their labels in consecutive order on the convex hull, followed by the labels for all clauses edges, i.e., they can be separated by a line through the convex point set. %
We now sketch how these edges are constructed.
}%
For every variable $x_i$, we add one edge to $E(G)$. One vertex can only be placed on the point $x_i^a$ and the other vertex can be placed on either $x_i^b$ in clockwise direction %
\NewText{from} $x_i^a$ or $x_i^{b'}$ in counterclockwise direction %
\NewText{from} $x_i^{a'}$ as depicted in \cref{fig:hardmat}a, %
encoding different truth assignments to the variable~$x_i$. If $x_i^{a}$ is connected to $x_i^{b}$, then this corresponds to $x_i$ being set to \textsc{False}. If it is connected to $x_i^{b'}$, then this corresponds to $x_i$ being set to \textsc{True}.
For every clause $C_j$, we add one edge to $E(G)$. One of its vertices can be placed on points $C_j^a$ or $C_j^{a'}$ and the other on $C_j^b$ or $C_j^{b'}$, leaving a total of four combinations to embed the edge. This is depicted by the dashed lines in \cref{fig:hardmat}b.
For every containment relation $x_i \in C_j$ or $\lnot x_i\in C_j$, we add one edge to~$E(G)$. One vertex can be placed on $xc_{i,j}^a$ or $xc_{i,j}^{a'}$, the other on $xc_{i,j}^b$ or $xc_{i,j}^{b'}$.%
The points with these labels are placed such that the line $xc_{i,j}^axc_{i,j}^b$ crosses only the line $x_i^ax_i^b$, but no other potential edge in the embedding of $G$. A symmetric description can be made for the straight line between $\lnot xc_{i,j}^a$ and $\lnot xc_{i,j}^b$ crossing only the straight line $x_i^ax_i^{b'}$. In the same way the line $xc_{i,j}^{a'}xc_{i,j}^{b'}$ only crosses the line $C_j^{a'}C_j^{b'}$ and one other straight line of the clause gadget in the embedding. Different literals belonging to the same clause cross a different second straight line of the potential embeddings of the clause gadget, as illustrated in  \cref{fig:hardmat}b. \NewText{We illustrate a full instance in \cref{fig:sub1} and a solution in \cref{fig:sub2}. The modification of the instance for trees can be seen in \cref{fig:sub3} and the solution in \cref{fig:sub4}.}%
\end{proofsketch}
\begin{proof}%
    We first construct an instance $\Instance = \InstanceLong$ of \ConvexListPointEmbedding where $G$ is a matching as follows (we discuss the modification for trees in the end).
\NewText{
We will represent variables, clauses, and variable-clause containment relations via edges in $G$.
We place the labels on the convex point set in such a way that labels of variable edges and variables of clause edges do not appear in alternating order. All the variable edges have their labels in consecutive order on the convex hull. Then, all the clause edges have their labels in consecutive order on the convex hull. In particular, we can draw some line through the convex point set that separates all the clause edges from all the variable gadgets. Which will be useful in the next step of the construction.

For every variable $x_i$, we add one edge to $E(G)$. One vertex can only be placed on the point $x_i^a$ and the other vertex can be placed on either $x_i^b$ in clockwise direction %
\NewText{from} $x_i^a$ or~$x_i^{b'}$ in counterclockwise direction %
\NewText{from} $x_i^{a'}$ as depicted in \cref{fig:hardmat}a, %
encoding different truth assignments to the variable~$x_i$. If $x_i^{a}$ is connected to $x_i^{b}$, then this corresponds to $x_i$ being set to \textsc{False}. If it is connected to $x_i^{b'}$, then this corresponds to $x_i$ being set to \textsc{True}.
For every clause $C_j$, we add one edge to $E(G)$. One of its vertices can be placed on points $C_j^a$ or $C_j^{a'}$ and the other on $C_j^b$ or $C_j^{b'}$, leaving a total of four combinations to embed the edge. This is depicted by the dashed lines in \cref{fig:hardmat}b.
For every containment relation $x_i \in C_j$ or $\lnot x_i\in C_j$, we add one edge to~$E(G)$. One vertex can be placed on $xc_{i,j}^a$ or $xc_{i,j}^{a'}$, the other on $xc_{i,j}^b$ or $xc_{i,j}^{b'}$. %
The points with these labels are placed such that the line $xc_{i,j}^axc_{i,j}^b$ crosses only the line $x_i^ax_i^b$, but no other potential edge in the embedding of $G$. A symmetric description can be made for the straight line between $\lnot xc_{i,j}^a$ and $\lnot xc_{i,j}^b$ crossing only the straight line $x_i^ax_i^{b'}$. In the same way the line $xc_{i,j}^{a'}xc_{i,j}^{b'}$ only crosses the line $C_j^{a'}C_j^{b'}$ and one other straight line of the clause gadget in the embedding. 

\begin{figure}
    \centering
	\includegraphics[page=3,scale=0.6]{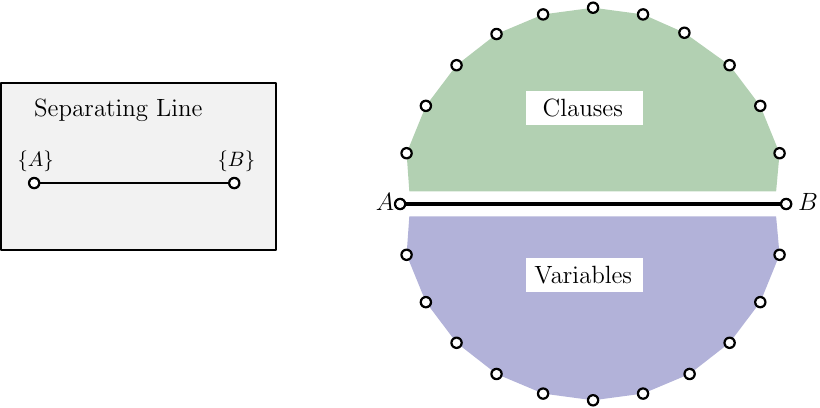}
    \caption{%
    \NewText{Line separating the clause edges from the variable gadgets.}}
	\label{fig:separate}
\end{figure}

To avoid the two cases in which $xc_{i_1,j}^{a'}$ and $xc_{i_1,j}^b$ or $xc_{i_1,j}^a$ and $xc_{i_1,j}^{b'}$ are edges of the straight line embedding, we add a separating line as depicted in \cref{fig:separate} and place all points corresponding to the construction of the variable gadgets on one side and points for clause gadgets on the other. The straight lines described at the start of this paragraph will then always cross the separating line. Consequently, they will never be part of a straight line embedding. %

Now in any realization of the graph $G$ each containment gadget need to either connect the two vertices on the variable side (and thus requires the clause gadget to represent the right truth assignment) or connect the two vertices on the clause side (and thus eliminates one of three possible spots on where to place the clause gadget).
We now argue correctness of the reduction.
}%

    If $\phi$ is satisfiable, then we embed every variable gadget in the described way that corresponds to its truth assignment in the satisfying assignment. For every containment gadget, we embed it on the variable side if the truth assignment of the variable matches the literal in the clause and on the clause side otherwise. Since every clause in $\phi$ is satisfied, at least on of the possible embeddings for the clause gadget will be uncrossed, allowing us to embed these edges. We conclude that $G$ can be embedded in a way that is admissible with $L$.

    Conversely assume $G$ can be embedded in a way that is admissible with $L$. We deduce for every variable in $\phi$ the truth value from the embedding of the corresponding variable gadget. Since every clause gadget has been embedded, we conclude that at least one possible embedding for every clause gadget exists which corresponds to the corresponding clauses being satisfied in $\phi$ by the construction of the containment gadgets.

    \proofsubparagraph{Modification for Trees.} 
    The following easy modification shows that the same proof strategy also works for bi-labeled trees in convex point sets. To $E(G)$ we add an edge  from the vertex with label set $\{A\}$ to every vertex with label set $\{x_i^{a}\}$ in the variable gadget and the vertex with labels $\{C_j^{a},C_j^{a'}\}$ in the clause gadget. For the vertices with label $\{xc_{i,j}^{a},xc_{i,j}^{a'}\}$ or $\{\lnot xc_{i,j}^{a},\lnot xc_{i,j}^{a'}\}$ in the containment gadget, we identify the vertex with these labels with the vertex with label $\{A\}$ meaning that they are the same vertex. The new, combined vertex has label set $L=\{A\}$. %
That way, $G$ is a connected graph. The additional and modified edges will, once embedded, not cross one another, since they are all incident to the same vertex. Additionally, in any admissible embedding of a bi-labeled matching, we can add all the additional edges without creating any crossings. Consequently, we can admissibly embed the constructed tree if and only if we can admissible embed the matching. 
This completes the proof.
\end{proof}

\NewText{We illustrate a full instance in \cref{fig:sub1} and a solution in \cref{fig:sub2}. The modification of the instance for trees can be seen in \cref{fig:sub3} and the solution in \cref{fig:sub4}.}%

\begin{figure}
    \centering

    \begin{subfigure}[b]{0.42\textwidth}
        \centering
        \includegraphics[page =2, width=\linewidth]{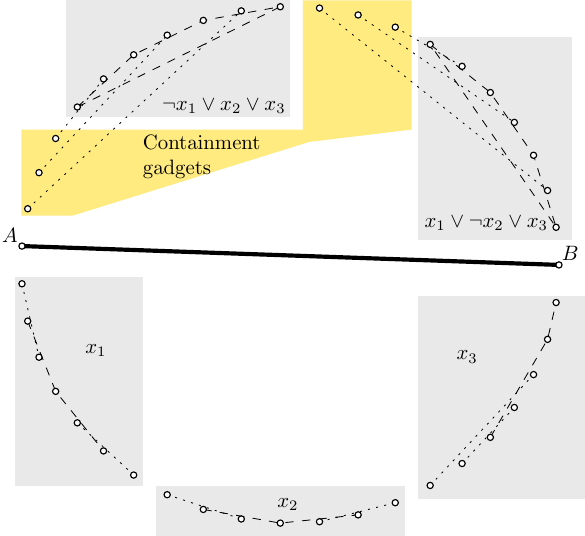}
        \caption{}%
        \label{fig:sub1}
    \end{subfigure}
    \hfill
    \begin{subfigure}[b]{0.42\textwidth}
        \centering
        \includegraphics[page = 3, width=\linewidth]{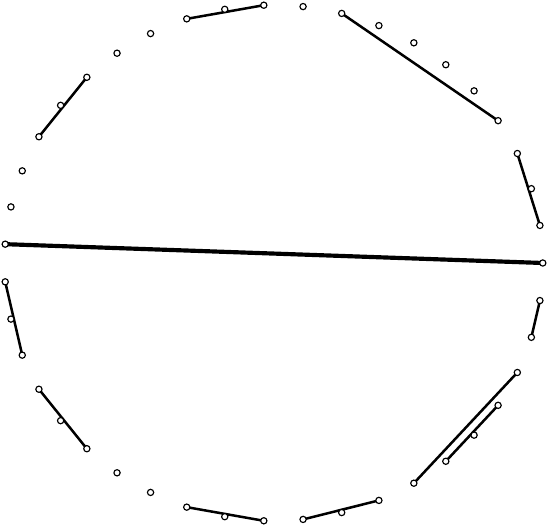}
        \caption{}%
        \label{fig:sub2}
    \end{subfigure}

    \vspace{0.5cm}

    \begin{subfigure}[b]{0.42\textwidth}
        \centering
        \includegraphics[page = 4, width=\linewidth]{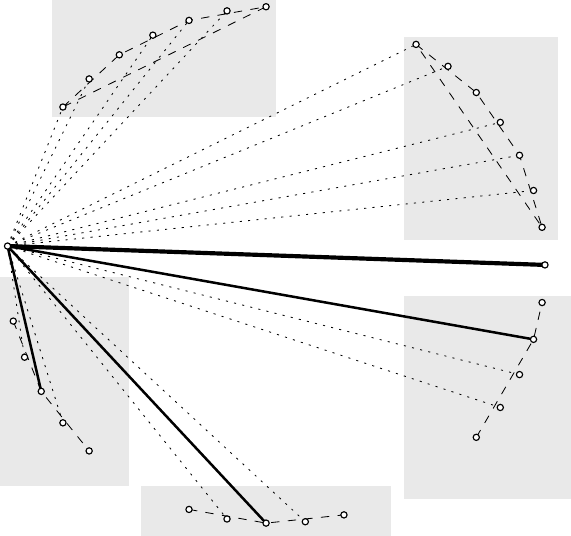}
        \caption{}%
        \label{fig:sub3}
    \end{subfigure}
    \hfill
    \begin{subfigure}[b]{0.42\textwidth}
        \centering
        \includegraphics[page = 5, width=\linewidth]{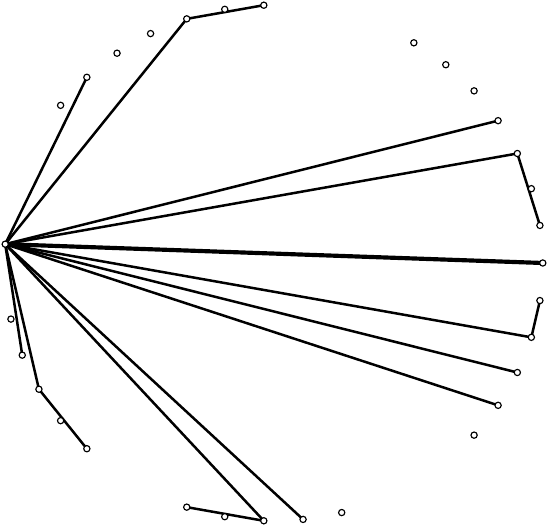}
        \caption{}%
        \label{fig:sub4}
    \end{subfigure}

    \caption{%
    \textbf{\textsf{(a)}} Instance of \ConvexListPointEmbedding with matchings resulting from $(\lnot x_1 \lor x_2 \lor x_3)\land (x_1 \lor \lnot x_2 \lor x_3)$. \textbf{\textsf{(b)}} A solution to \ConvexListPointEmbedding with matchings resulting from the truth assignment $x_1=\textsc{True}$, $x_2=\textsc{False}$, $x_3=\textsc{True}$. 
    \textbf{\textsf{(c)}} The instance from \cref{fig:sub1} modified for trees. \textbf{\textsf{(d)}} The solution from \cref{fig:sub2} modified for trees.}
    \label{fig:all}
\end{figure}
}

\subsection{Efficient Algorithms for Connected Graphs}
\label{sec:convex-embedded}
In this section, we show that \ConvexListPointEmbedding can be solved in polynomial time for connected graphs $G$ with a given \emph{combinatorial embedding}, an equivalence class of planar drawings of~$G$ that have the same rotation system, i.e., counter-clockwise cyclic order
of incident edges around each vertex, and the same outer face.
Afterwards we remove the assumption on the combinatorial embedding and turn the algorithm into an \FPT-algorithm in the maximum degree of $G$.
Our algorithm heavily relies on the block-cut-tree:
A \emph{block} of $G$ is a maximal connected subgraph without a cutvertex, i.e., every block of~$G$ is either a maximal two-connected subgraph or a bridge (i.e., separating edge) otherwise.
The \emph{block-cut-tree} (or \emph{BC-tree}) $\mathcal T$ of $G$ is the unique tree
whose nodes are the blocks and cutvertices of $G$ and each
edge connects a cutvertex node to a block node so that each node of a cutvertex is
adjacent to exactly all those nodes of blocks that the cut vertex is contained in $G$; see~\cref{fig:bctree}.
Consider $\mathcal T$ to be rooted at some block node.
Given a node $\nu$ of $\mathcal T$, the \emph{pertinent graph} of $\nu$ is the
subgraph of $G$ induced by the edges in all blocks contained in the subtree of $\mathcal T$
rooted at $\nu$.
For convenience, we also define a pertinent graph for each vertex of $G$:
for each cutvertex $v$ of $G$, the \emph{pertinent graph} of $v$ is the pertinent graph of the node of $\mathcal T$ corresponding to $v$ and
for each non-cutvertex $v$ of $G$, we define its \emph{pertinent graph} to be the graph $(\{v\},\emptyset)$, which contains only $v$.
We are now ready to state the result that we want to prove.

\begin{restatable}\restateref{thm:convexDP}{theorem}{theoremConvexDP}
\label{thm:convexDP}The following hold:
    \begin{itemize}
        \item There is an algorithm \Ai that, given an instance $\Instance = \InstanceLong$
        of \ConvexListPointEmbedding where $G$ is
        connected and equipped with a combinatorial embedding~$\mathcal G$, computes an $L$-realization
        of $G$ that respects $\mathcal G$ (or correctly determines that such a
        realization does not exist) in $\mathcal O(n\cdot \ell\cdot s)$ time and $\mathcal O(n\cdot s)$ space; and
        \item there is an algorithm \Aii that, given an instance $\Instance = \InstanceLong$
        of \ConvexListPointEmbedding where $G$ is
        connected, computes an $L$-realization of $G$ (or correctly determines that such a
        realization does not exist) in $\mathcal O(\Delta!\cdot n\cdot \ell\cdot s)$
        time (which is \FPT\ in $\Delta$) and $\mathcal O(n\cdot s)$ space.
    \end{itemize}
\end{restatable}

\begin{figure}
    \centering
    \includegraphics{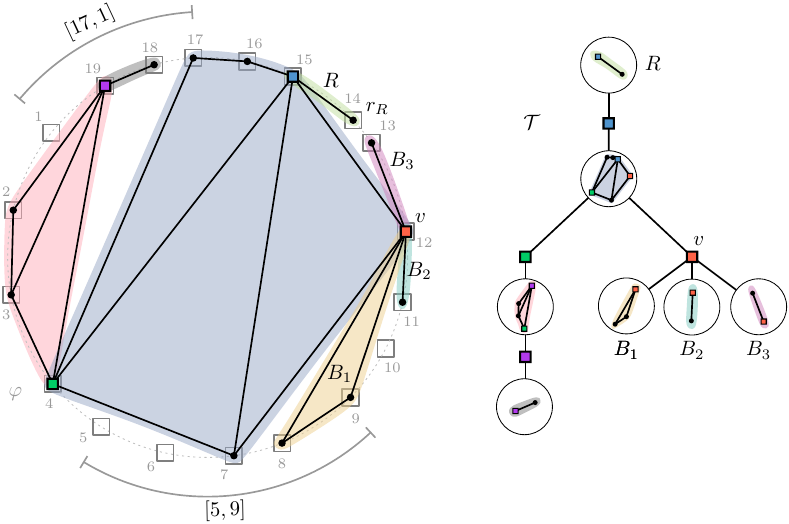}
    \caption{A convex (outerplanar) realization together with the BC-tree of the correpsonding graph; shaded regions correspond to blocks.
    The induced child block order of $v$ is $B_1, B_2, B_3$.
    }
    \label{fig:bctree}
\end{figure}

To prepare for the proof of \cref{thm:convexDP}, we first observe that solutions
obtained by the claimed algorithms are necessarily outerplanar and proceed by stating
some simple but useful properties of %
combinatorial embeddings \NewText{of such graphs}, some of
which are illustrated in~\cref{fig:bctree}.

\begin{observation}
\label{obs:solutionOuterplanar}
    Every realization constructed by \Ai or \Aii is outerplanar.\qed
\end{observation}

\begin{observation}
\label{obs:augmentedcycle}
    In every outerplanar combinatorial embedding of a $2$-connected graph, the outer face is bounded by a simple cycle and all remaining edges are chords of this cycle.\qed
\end{observation}

\begin{corollary}
\label{cor:uniqueembedding}
    Outerplanar combinatorial embeddings of $2$-connected graphs are unique up to reflection.\qed
\end{corollary}

\begin{observation}
\label{obs:cutvertices}
Let $G$ be a connected graph, let $v$ be a cutvertex of $G$, let $B$ be a block
incident to $v$, and let $\mathcal G$ be an outerplanar combinatorial embedding of $G$.
In the cyclic order of incident edges of $v$ in $\mathcal G$, the edges that belong to
$B$ are consecutive.\qed
\end{observation}

Let $\mathcal G$ be an outerplanar combinatorial embedding of a connected graph
$G$, let $v$ be a cutvertex of $G$, let $\mathcal T$ \NewText{be} the BC-tree of $G$, which we
consider to be rooted at some block node, and let $B$ be the block that corresponds to
the parent of $v$ in $\mathcal T$.
In view of \Cref{obs:cutvertices}, in $\mathcal G$ the blocks of $v$ appear around $v$
in a well-defined counterclockwise cyclic order $B, B_1,B_2, \dots, B_k$.
We say that $B_1,B_2, \dots, B_k$ is the \emph{child block order} of $v$ \emph{induced} by $(\mathcal T, \mathcal G)$.

We are now ready to describe the proof of \Cref{thm:convexDP}.

\begin{proof}[Proof of \Cref{thm:convexDP}]
\Ai and \Aii are very similar; essentially \Aii corresponds to \Ai, but some of its steps
are repeated for different ways to choose the combinatorial embedding.
Naturally, this means that the correctness and runtime arguments for both algorithms
are also very similar.
For this reason, we discuss both algorithms within the same theorem/proof environment. 
We use the following notation to clearly mark the few situations where the (analysis of
the) algorithms differ: \algs{Description that only applies to \Ai}{Description that only
applies to \Aii}.
Readers that are specifically interested in the discussion about \Ai
(\Aii) can therefore simply skip over
the second (first) entries of these text constructs.

We begin by introducing some notation.
    Let $\varphi = (x_0, x_1, \dots, x_{s-1}, x_0)$ be the (given) circular ordering of the
    convex set $S$ obtained by visiting its points in counterclockwise fashion.
    For brevity, we identify each point $x_i$ with its index $i$.
    For $p, q \in \{0, \dots, s-1\}$, we use $[p, q]$ to denote the linear subsequence of $\varphi$
    that starts with $p$ and ends with $q$.
    As a convention, calculations involving these indices are performed modulo $s$; e.g., $p+1$
    denotes the counterclockwise successor of $p$ in $\varphi$ and $[p,p-1]$
    denotes the linearization of $\varphi$ that starts at $p$; see \cref{fig:bctree} for examples.

In this paragraph, we discuss the \emph{preprocessing step} of the algorithms, which
starts by computing the BC-tree $\mathcal T$ of $G$, which can be done in $\mathcal O(n)$
time~\cite{HT.EAG.1973}.
\algs{
\Ai then tests if the given combinatorial embedding $\mathcal G$ is outerplanar, which is easy to do
in $\mathcal O(n)$ time; if $\mathcal G$ is not outerplanar, we report that there is
no solution (which is correct due to \cref{obs:solutionOuterplanar}).}{
\Aii then attempts to construct an outerplanar combinatorial embedding $\mathcal G_B$
for each block $B$, which can be done in $\mathcal O(n)$ total time~\cite{Mit.LAR.1979};
if some block does not admit such an embedding, we report that
there is no solution (which is correct due to \cref{obs:solutionOuterplanar}).}
The overall runtime of the preprocessing step is thus $\mathcal O(n)$.

    We then employ a dynamic programming strategy that, intuitively speaking,
    computes a ``shortest possible'' realization of our graph by greedily combining
    recursively computed ``shortest possible'' realizations of suitably defined subgraphs.
    As typical with dynamic programming, we first describe an algorithms that only decides whether a solution exists; in the end, we make it constructive by employing
    the usual back-linking strategy.
    
    To set up the necessary notation, we consider $\mathcal T$ to be rooted at a block
    node of degree at most $1$ and use $R$ to denote the corresponding block.
    Let $r_R$ be a vertex of $R$ that is not a cutvertex of $G$.
    For all blocks $B \neq R$, we define $r_B$ to be the cutvertex of $G$ that
    corresponds to the parent of $B$ in $\mathcal T$.
    For each block $B$ (including $R$) and each %
    \NewText{point} $p \in S$, we define $T[B, p]$ to be the first
    \NewText{point} $q \in [p,p-2]$ such that the pertinent graph of $B$ admits an $L$-realization $\Gamma$ 
    \algs{that respects $\mathcal G$ %
    and}{%
    \textbf{-}}
    where all vertices except for $r_B$ are placed in $[p, q]$ and $r_B$ is placed outside
    of $[p, q]$ (i.e., in $[q+1,p-1]$); in case no such %
    \NewText{point} $q$ exists, we define $T[B, p] = \infty$.
    If $T[B,p]$ is finite, we say that $\Gamma$ (and any other realizations with the same properties) is a realization \emph{corresponding} to $T[B,p]$.
    It is easy to observe that the exact position of $r_B$ is irrelevant -- its valid
    positions are freely interchangable:
    
    \begin{ob}
    \label{cl:root-movable}
        Let $\Gamma$ be a realization corresponding to some entry $T[B,p]=q$.
        Moving $r_B$ from its %
        \NewText{placement} in $\Gamma$ to some other %
        \NewText{point from} $[q+1,p-1]$ that belongs to $L(r_B)$ results in another realization corresponding to
        $T[B,p]$.\claimqedhere
    \end{ob}
    
    Note that there is an $L$-realization of $G$ \algs{that respects $\mathcal G$%
    }{%
    \textbf{-}} if and only if $T[R, p+1]$ is finite for at least one
    \NewText{point} $p \in L(r_R)$.
    Given all entries of $T$, this condition can be checked in $\mathcal O(|L(r_R)|)\subseteq \mathcal O(\ell)$ time,
    which corresponds to the \emph{postprocessing step} of the algorithms.

It remains to discuss how to compute the entries of $T$.
The following algorithmic subroutine will prove to be a valuable tool in this regard.    
    
    \begin{claim}
    \label{cl:children}    
    \label{cl:childrenBody*}
       Let $v$ be a vertex of $G$ and let $p \in S$. 
       We can compute the first %
       \NewText{point} $q \in [p, p-2]$ so that the
       pertinent\footnote{Recall that we defined the pertinent graph of a
       cutvertex to be the pertinent graph of the corresponding node of $\mathcal T$, whereas
       the pertinent graph of a non-cutvertex consists of only that vertex; c.f.~\cref{sec:preliminaries}.} graph $G_v$ of $v$
       admits an $L$-realization $\Gamma$
       in $[p, q]$ \algs{that respects $\mathcal G$%
    }{%
    \textbf{-}} where
       \begin{enumerate}
       \item \label{cl:children:prop1} no edge $ab$ of $G_v$ is realized in a way where $a$
       lies before $v$ and $b$ lies after $v$ in $[p, q]$; and
       \item \label{cl:children:prop2} \algs{%
       if $v$ is a cutvertex of $G$, then,
       for each pair of blocks $B_a,B_b$ of $G_v$ that are incident to $v$ and
       where $B_a$ appears before $B_b$ in the child block order of $v$ induced by $(\mathcal T,\mathcal G)$,
       no vertex of the pertinent graph of $B_b$ that is not $v$
       appears before some vertex of the pertinent graph of $B_a$ that is not $v$}{%
       \textbf{-}}
       \end{enumerate}
        (or return $\infty$ in case such %
        \NewText{point} $q$ does not exist) in \algs{$\mathcal O(|L(v)|\cdot j)$%
        }{$\mathcal O(j!\cdot |L(v)|\cdot j)$%
        } time, where $j$ is the number of children of $v$ in $\mathcal T$ if $v$ is a cutvertex and $j=1$ otherwise, assuming, in case $v$ is a cutvertex, 
        constant-time access to all entries $T[B', p']$ for all $p' \in S$ and all blocks $B'$ that
        are children of $v$ in $\mathcal T$.
    \end{claim}
    
    \begin{claimproof}
    In case $v$ is not a cutvertex of $G$, the desired %
    \NewText{point} $q$ simply corresponds
    to the first index in $[p,p-2]$ that occurs in $L(v)$ (or $\infty$ if no such %
    \NewText{point} exists).
    This case can thus be handled in $\mathcal O(|L(v)|)$ time.%
    
    So assume that $v$ is a cutvertex of $G$.
    Further, assume that $q$ is finite.
    We say that $\Gamma$ and any other realizations with the same properties is a realization \emph{corresponding} to $q$.%
    \algs{%
    Let $B_1,\dots,B_j$ the child block order of $v$ induced by $(\mathcal T,\mathcal G)$. We will
    describe a strategy that directly computes the desired %
    \NewText{point} $q$.}{%
    Let
    $B_1,\dots,B_j$ be some arbitrary permutation of the blocks that correspond to
    children of $v$ in $\mathcal T$. We will describe a strategy that determines a %
    \NewText{point} that is $q$ %
    assuming there exists a corresponding realization where for each pair of blocks
    $B_a,B_b$ with $1\le a<b\le j$, no vertex of the pertinent graph of $B_b$ that is not
    $v$ appears before some vertex of the pertinent graph of $B_a$ that is not $v$ (this
    condition basically corresponds to Property~\ref{cl:children:prop2} for \Ai, but applied
    to the considered permutation instead of the child block order induced by $(\mathcal T,\mathcal G)$).
    We will repeat our strategy $j!-1$ times -- once for each other permutation
    of $B_1,\dots,B_j$, which incurs a factor of $\mathcal O(j!)$ in the runtime. Among all the thereby determined %
    \NewText{points}, the first that is
    contained in $[p,p-2]$ is the correct %
    \NewText{point}~$q$ since, due to
    \cref{obs:cutvertices}, the described \emph{block ordering property} necessarily has to be
    satisfied for some block permutation in every realization corresponding to $q$.
    So from now on, assume without loss of generality that there is a realization corresponding to $q$ in which the block ordering property is satisfied for $B_1,\dots,B_j$.}

Consider a %
\NewText{point} %
$p_v\in L(v)$ located in
$[p,p-2]$ (such a %
\NewText{point} has to exist since the desired
realization exists by assumption).
We will describe a strategy that determines a %
\NewText{point}
that is~$q$ assuming that there is a corresponding
realization \algs{%
\textbf{-}}{in
which the block ordering property is satisfied for
$B_1,\dots,B_j$ %
and}  where $v$ is placed at
$p_{v}$.
We repeat this strategy for all other %
\NewText{points} in
$L(v)$ that are located in $[p,p-2]$, which incurs a factor of $|L(v)|$ in the runtime.
Among all the thereby obtained %
\NewText{points}, the first that
is contained in $[p,p-2]$ is the correct %
\NewText{point} $q$
(since the desired realization exists by assumptions and
also uses some of these %
\NewText{points} for $v$). %
So from now on, assume without loss of generality that
there is a realization corresponding to $q$  \algs{%
\textbf{-}}{in
which the block ordering property is satisfied for
$B_1,\dots,B_j$ %
and} where $v$ is placed at $p_{v}$.

We define $q_0=p-1$. For $i\in \{1,\dots,j\}$,  we define
$q_i$ to be the first %
\NewText{point} in $[p,p-2]$ such that
there is a realization $\Gamma_i$ corresponding to $q$
\algs{%
\textbf{-}}{in
which the block ordering property is satisfied for
$B_1,\dots,B_j$ %
and} where $v$ is placed at $p_{v}$
and the last %
\NewText{point on which} %
a vertex in the pertinent graph
of $B_i$ that is not $v$ \NewText{is placed} is $q_i$ (intuitively speaking,
$q_i$ is the endpoint of a shortest possible
(sub)realization of the first $i$ blocks in the
considered realizations corresponding to $q$, ignoring the
common parent $v$).
We say that $\Gamma_i$ and every other realization with
the same properties is a realization \emph{corresponding}
to $q_i$.
Observe that, by definition, the first %
\NewText{point} among
$q_j$ and $p_{v}$ in $[p,p-2]$ is the desired value $q$.
It thus remains to describe how to determine $q_j$, which
we will do in an iterative greedy-like fashion.

\begin{figure}
    \centering
    \includegraphics[page=1]{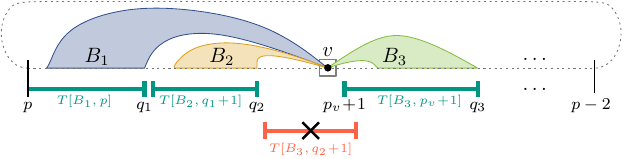}
    \caption{
    The greedy strategy used in the proof of \cref{cl:children}.
    The shaded regions represent the pertinent graphs of $B_1,\dots,B_j$.
    }
    \label{fig:dpClaim}
\end{figure}

Suppose, inductively, we have already obtained the
\NewText{point} $q_i$ for some $i\in \{0,\dots,j-1\}$ (note that
obtaining $q_0$ is trivial as $q_0 = p - 1$ by definition).
To obtain $q_{i+1}$, we proceed as follows (for illustrations, see \cref{fig:dpClaim}).
Since, by assumption, a realization corresponding to $q_i$ exists,
the %
\NewText{point} $T[B_{i+1}, q_i+1]$ is finite and contained in $[q_i+1,p-2]$.
If $p_{v}$ is not contained in $[q_i+1, T[B_{i+1}, q_i+1]]$, then $q_{i+1}=T[B_{i+1}, q_i+1]$ by definition of the involved %
\NewText{points} and \cref{cl:root-movable}.
Otherwise, due to Property~\ref{cl:children:prop1} and by the assumption that a realization corresponding to $q_i$ exists, we have $q_{i+1}=T[B_{i+1},p_{v}+1]$.%

Obtaining $q_{i+1}$ from $q_i$ takes $\mathcal O(1)$ time (recall that we assume constant-time access to all entries $T[B', p']$ for all $p' \in S$ and all blocks $B'$ that
        are children of $v$ in $\mathcal T$).
        Thus, obtaining~$q_j$ takes $\mathcal O(j)$ time (for some fixed %
        \NewText{point for} $v$) \algs{%
        \textbf{-}
        }{(and some fixed permutation of the child blocks of $v$ in $\mathcal T$)%
        }.
Overall, we therefore obtain a total runtime of \algs{$\mathcal O(|L(v)|\cdot j)$}{$\mathcal O(j!\cdot |L(v)|\cdot j)$}.
    \end{claimproof}

    To compute an entry $T[B, p]$, we consider two cases depending on the type of $B$.
    In both cases, we explicitly describe how to realize $B$. Pertinent graphs of children
    of $B$ in $\mathcal T$ will be realized recursively by appealing to \cref{cl:children};
    the base case of the recursion corresponds to the situation where $B$ is a leaf of
    $\mathcal T$, in which case the pertinent graph of $B$ is just $B$.
    To justify the assumption in \cref{cl:children}, we employ the usual memoization strategy; alternatively, one could process the entries in a bottom-up fashion (starting with the entries whose blocks correspond to leaves of the BC-tree).

    \proofsubparagraph{Case 1:}
    $B$ is a single edge $r_Bv$.
    Let $q$ be the result of applying \cref{cl:children} to $v$ and $p$.
    If $q=\infty$ or if $q$ is finite, but $[q+1,p-1]$ contains no %
    \NewText{point from} $L(r_B)$,
    no realization corresponding to $T[B,p]$ is possible, so we set $T[B,p]=\infty$.
    Otherwise, note that, due to Property~(\ref{cl:children:prop1}) from
    \Cref{cl:children}, the realization that is implicitly constructed by \Cref{cl:children} can
    be augmented by the edge $r_Bv$ for some %
    \NewText{placement for} $r_B$ %
    \NewText{on a point from} $[q+1,p-1]$ and
    $L(r_B)$ to obtain a realization corresponding to $T[B,p]=q$, see \cref{fig:dpCases}(a).
    \algs{(Notably, due to Property~(\ref{cl:children:prop2}) from
    \Cref{cl:children}, this augmentation correctly places the block $B$ after the last / before the first block of the child block order of $v$ induced by $\mathcal T,\mathcal G$.)}{\textbf{-}}. %
    Case~1 can thus be handled in \algs{$\mathcal O(|L(v)|\cdot j)$}{$\mathcal O(j!\cdot |L(v)|\cdot j)$} time, where $j$ is the number of children of~$v$ in $\mathcal T$ if $v$ is a cutvertex and $j=1$ otherwise.
    Note that we can express this time bound as
    \algs{$\mathcal O(\sum_{u\in V(B)\setminus \{r_B\}}\ell \cdot (\mathrm{deg}_G(u)))$}{$\mathcal O(\sum_{u\in V(B)\setminus \{r_B\}}(\mathrm{deg}_G(u)-1)!\cdot \ell\cdot (\mathrm{deg}_G(u)))$},
    which matches the running time we will obtain for Case~2.
    \hfill$\diamond$

\begin{figure}
    \centering
    \includegraphics[page=2]{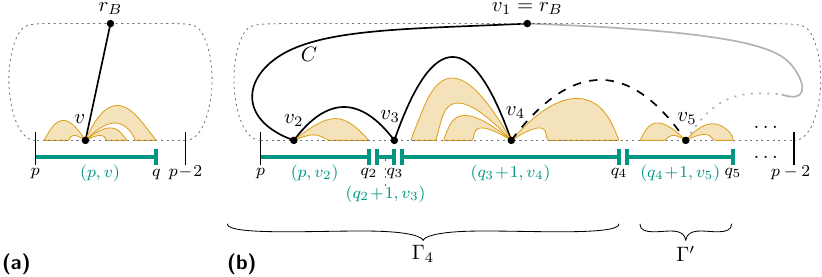}
    \caption{
    \textbf{(a)} Case 1. \textbf{(b)} Case 2. The tuples indicate the inputs for applications of  \cref{cl:children}.
    }
    \label{fig:dpCases}
\end{figure}

    \proofsubparagraph{Case 2:}
    $B$ is not a single edge, i.e., $B$ is a maximal $2$-connected subgraph of $G$.
    \algs{The following strategy directly computes the desired value $T[B,p]$.}{The following strategy describes how to find a %
    \NewText{point} that is $T[B,p]$
    assuming that the corresponding realization uses the precomputed (in the preprocessing step) combinatorial
    embedding $\mathcal G_B$ for $B$.
    We then repeat the strategy once more to obtain a %
    \NewText{point} that is $T[B,p]$
    assuming that the corresponding realization uses the reflection of $\mathcal G_B$ for $B$, which incurs a (vanishing) factor of $2\in \mathcal O(1)$ in the runtime.
    Due to \cref{cor:uniqueembedding}, the correct value of $T[B,p]$ then is the first
    \NewText{point from} $[p,p - 2]$ corresponding to one of the two obtained %
    \NewText{points} (or $\infty$
    if both of them are $\infty$).
    So from now on, assume without loss of generality that, if there is a realization corresponding to $T[B,p]$, then there also is a realization corresponding to $T[B,p]$ that uses the embedding $\mathcal G_B$ for $B$.}
    
    In view of \Cref{obs:augmentedcycle}, 
    \algs{the combinatorial embedding of $B$ in $\mathcal G$}{the
    precomputed embedding $\mathcal G_B$} is a counterclockwise cycle
    $C = (r_B = v_1, v_2, \dots, v_j, v_1)$, possibly augmented with some interior chords $I$.
    Note that in any $L$-realization of $G_B - I$, where $G_B$ is the pertinent graph of 
$B$, the cycle $C$ has to be represented as a convex polygon.
    By convexity, the chords $I$ can be inserted into any such realization, therefore, we 
will assume without loss of generality that $I = \emptyset$.
    
    For $i \in \{1, \dots, j\}$, let $H_i$ denote the graph consisting of the path
    $v_1, \dots, v_i$, together with the union of the pertinent graphs of $v_2, \dots, v_i$.
    We define $q_1 = p-1$.
    For $i \in \{2, \dots, j\}$, we define $q_i$ to be the first %
    \NewText{point} in $[p, p-2]$ such 
that $H_i$ has an $L$-realization $\Gamma_i$ \algs{that respects $\mathcal G$}{where the embedding of $B$ is $\mathcal G_B$} and where (i) 
all vertices except for $r_B$ ($=v_1$) are placed in $[p, q_i]$, (ii) $r_B$ is placed in
    $[q_i+1,p-1]$, (iii) among the vertices $v_2, \dots, v_i$, the vertex $v_i$ is the last in
    $[p, q_i]$ and (iv) there is no edge $ab$ of $H_i$ where $a$ lies before $v_i$ and $b$
    lies after $v_i$ in $[p, q_i]$; if no such $q_i$ exists, we define $q_i = \infty$.
    If $q_i$ is finite, we say that $\Gamma_i$ and every other realization with the same properties is an $L$-realization 
\emph{corresponding} to $q_i$. 
        
    Suppose, inductively, we have already determined $q_i$, $i \in \{1, \dots, j-1\}$ 
(note that obtaining~$q_1$ is trivial as $q_1 = p - 1$ by definition).
    To obtain $q_{i+1}$, we proceed as follows.
    If $q_i=p-2$, we set $q_{i+1}=\infty$ since in this case the interval $[p,p-2]$ is 
already fully occupied by every realization of $H_i$ corresponding to $q_i$ and there is
    no way to place $v_{i+1}$ anymore.
    So assume that $q_i\in [p,p-3]$.
    Let $q'$ be the result of applying \Cref{cl:children} to $v_{i+1}$ and $q_i + 1$.
    If $q'=\infty$ or if $q'$ is finite but does not lie in $[q_i+1,p-2]$ or if $q'$ is finite
    and lies in $[q_i+1,p-2]$ but $[q'+1,p-1]$ contains no %
    \NewText{point} in $L(r_B)$ (=$L(v_1)$),
    we set $q_{i+1}=\infty$ as in this case the desired realization of $H_{i+1}$ does 
not exist.
    So assume otherwise, i.e., $q'$ is finite, lies in $[q_i+1,p-2]$, and $[q'+1,p-1]$ 
contains a %
\NewText{point from} $L(r_B)$ (=$L(v_1)$).
    Let $\Gamma_i$ be an $L$-realization of $H_i$ corresponding to~$q_i$ where 
$r_b$ ($=v_1$) lies in $[q'+1,p-1]$ (which exists by \cref{cl:root-movable}).
    Let $\Gamma'$ be the $L$-realization of the pertinent graph of $v_{i+1}$ that is 
(implicitly) constructed by \Cref{cl:children}.
    Note that $\Gamma_i$ and $\Gamma'$ do not intersect by definition.
    Further, by Property~(\ref{cl:children:prop1}) from \Cref{cl:children} and 
Property~(iv) from the previous paragraph, we can add the edge $v_iv_{i+1}$ to the union of $\Gamma_i$ and $\Gamma'$
    without introducing crossings, thereby creating an $L$-realization of $H_{i+1}$
    \algs{that respects $\mathcal G$ (notably, due to Property~(\ref{cl:children:prop2}) from
    \Cref{cl:children}, this construction correctly places the edge $v_iv_{i+1}$ of $B$ after the last / before the first block of the child block order of $v_{i+1}$ induced by $\mathcal T,\mathcal G$.)}{\textbf{-}}, see \cref{fig:dpCases}(b).
    It is easy to see that this realization fullfills Properties~(i)-(iv).%
    Further, a simple exchange argument (note that we have employed a greedy-like strategy) shows that $q'$ is the first %
    \NewText{point} in $[p,p-2]$
    for which such a realization exists and, thus, we obtain $q_{i+1}=q'$.
In the end, if $q_j$ is finite, we can add the final edge $v_jv_1$ of $C$ to the
    realization of $H_j$ corresponding to $q_j$ without introducing crossing due to
    Property~(iv) and, therefore, we obtain $T[B,p]=q_j$.

The runtime for obtaining $q_{i+1}$ from $q_i$ is dominated by the application of \Cref{cl:children}, which takes
\algs{$\mathcal O(|L(v_{i+1})|\cdot j_{v_{i+1}})$}{$\mathcal O(j_{v_{i+1}}!\cdot |L(v_{i+1})|\cdot j_{v_{i+1}})$} time, where $j_{v_{i+1}}$ is the number of children of $v_{i+1}$ in $\mathcal T$ if $v_{i+1}$ is a cutvertex and $j_{v_{i+1}}=1$ otherwise.
The total time for obtaining $q_j$ can thus be bounded by the same expression we obtained obtained for Case~1:  \algs{$\mathcal O(\sum_{u\in V(B)\setminus \{r_B\}}\ell \cdot (\mathrm{deg}_G(u)))$}{$\mathcal O(\sum_{u\in V(B)\setminus \{r_B\}}(\mathrm{deg}_G(u)-1)!\cdot \ell\cdot (\mathrm{deg}_G(u)))$}.
    \hfill$\diamond$\\
    
    Note that for two distinct blocks $B,B'$, we have $V(B)\setminus \{r_B\}\cap V(B')\setminus \{r_B'\}=\emptyset$.
    We obtain that for some fixed %
    \NewText{point} $p\in S$, the total time to obtained the %
    \NewText{points} $T[B,p]$ for all blocks $B$ is \algs{$\mathcal O(\sum_B\sum_{u\in V(B)\setminus \{r_B\}}\ell \cdot (\mathrm{deg}_G(u)))\subseteq \mathcal O(n \cdot \ell)$ (by the handshaking lemma)}{$\mathcal O(\sum_B\sum_{u\in V(B)\setminus \{r_B\}}(\mathrm{deg}_G(u)-1)!\cdot \ell\cdot (\mathrm{deg}_G(u)))\subseteq \mathcal O(\Delta!\cdot n\cdot \ell)$}.
    By multiplying with the number $s$ of %
    \NewText{points in} $S$, we obtain a total runtime of
    \algs{$\mathcal O(n\cdot \ell\cdot s)$}{$\mathcal O(\Delta!\cdot n\cdot \ell\cdot s)$} for computing all entries of $T$, which dominates the time bounds of $\mathcal O(n)$ and $\mathcal O(\ell)$ for the preprocessing and postprocessing step, respectively.
    Table $T$ has $\mathcal O(n\cdot s)$ entries.
    For the back-linking strategy (to make the algorithm constructive), for each entry $T[B,p]=q$, we also  store, for each $u\in V(B)\setminus\{r_B\}$, the %
    \NewText{point for} $u$ \algs{\textbf{-}}{and the permutation of the child blocks of $u$} as used in the realization corresponding to $q$.
    For some fixed %
    \NewText{point} $p\in S$, the total size of the back-links for all blocks $B$ is
    \algs{
    $\mathcal O(\sum_B\sum_{u\in V(B)\setminus \{r_B\}}1)\subseteq \mathcal O(n)$}
    {
$\mathcal O(\sum_B\sum_{u\in V(B)\setminus \{r_B\}}\mathrm{deg}_G(u))\subseteq \mathcal O(n)$ (by the handshaking lemma)
    }.
    By multiplying with the number $s$ of %
    \NewText{points in} $S$, we obtain the claimed total space consumption of $\mathcal O(n\cdot s)$.
    \qed
\end{proof}

\section{Embeddings on General Point Sets}
\label{sec:general}
\Cref{thm:convexDP} heavily relies on the fact that the points lie in convex position.
In this section, we drop the convexity requirement and study \ListPointEmbedding in its full generality.
We first show in \Cref{sec:general-hardness-path} that \ListPointEmbedding becomes \NP-complete even when restricted to bi-labeled paths or cycles. %
We then show in \Cref{sec:general-vc} that \ListPointEmbedding is fixed parameter tractable in $\vcn(G)$ for bi-labeled graphs and justify our restriction to bi-labeled graphs in \Cref{sec:parameterized-algorithms-star,sec:tri-labeled-vertex-cover-2}. %

\subsection{Embedding Bi-Labeled Paths and Cycles Becomes Hard}
\label{sec:general-hardness-path}

We first show that \ListPointEmbedding is \NP-complete even for bi-labeled paths and cycles.
To this end, we reduce from the \NP-complete problem \textsc{Planar Monotone 3-SAT}~\cite{BergKhosravi2010}.
In this variant of \textsc{3SAT}, the input consists of (i) a boolean formula $\phi$ where each clause contains only positive (i.e., un-negated) or only negative (i.e., negated) literals and (ii) a rectilinear representation of the incidence graph of $\phi$ where all variables lie on a horizontal line $h$, all positive clauses lie above $h$, and all negative clauses lie below $h$; see \Cref{fig:pm3sat} for an example.

\begin{figure}
\begin{subfigure}{\textwidth}
    \centering
    \includegraphics[page=6]{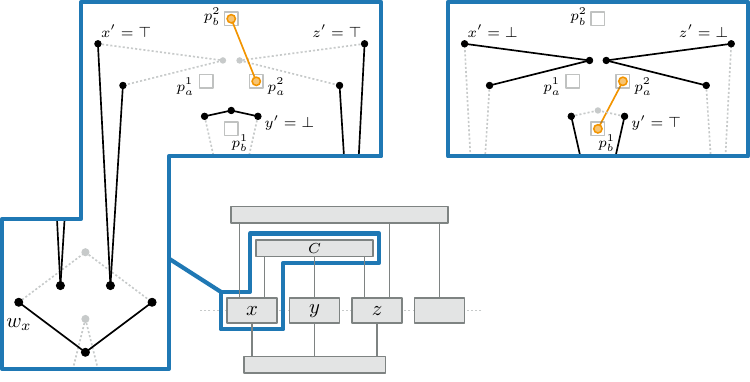}
    \subcaption{}
    \label{fig:pm3satA}
\end{subfigure}%
\newline
\begin{subfigure}{\textwidth}
    \centering
    \includegraphics[page=7]{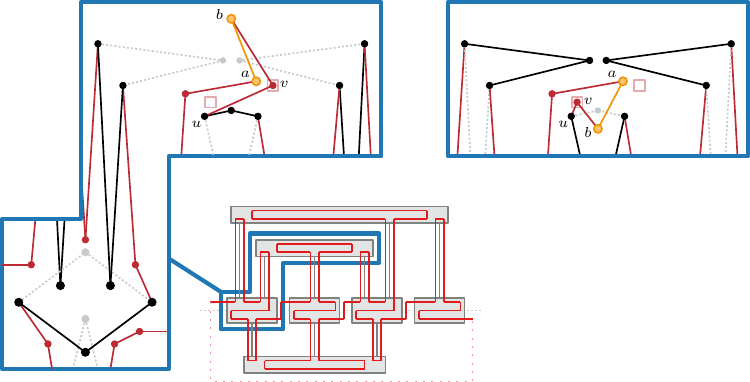}
    \subcaption{}
    \label{fig:pm3satB}
\end{subfigure}%
\caption{
\textbf{\textsf{(a)}}
A rectilinear representation of the incidence graph of $\phi$ (middle).
Highlighted on the left is the part of the constructed \ListPointEmbedding instance corresponding to the variable $x$ (assigned to true) and clause $C$.
The chosen state for each wedge is shown as black edges, the second state as dashed edges.
The chosen realization of the edge $ab$ in the clause gadget is highlighted in orange.
The right shows an additional possible realization of $ab$, corresponding to a different truth assignment.
\textbf{\textsf{(b)}}
The augmentation of the instance to a path.
The route of the path is highlighted in the incidence graph; additional vertices and edges of the path in the \ListPointEmbedding-instance are drawn in red.
The two admissible points of the clause gadget on which $v$ may be placed are shown as red~squares.
}
\label{fig:pm3sat}
\end{figure}

\begin{restatable}\restateref{thm:hardness-paths}{theorem}{theoremHardnessPath}
\label{thm:hardness-paths}
  \ListPointEmbedding is \NP-complete even for bi-labeled paths and cycles.
\end{restatable}
\begin{proofsketch}
Starting with the rectilinear representation of the incidence graph of $\phi$, we construct an equivalent instance $\InstanceLong$ of \ListPointEmbedding by replacing variables and clauses with gadgets as follows.
We first present the general approach and subsequently show how the graph of the resulting instance can be augmented to form a path (respectively a cycle).

The basic building blocks of our reduction are called \emph{wedges}. 
A wedge consists of a bi-labeled degree-2 vertex that is adjacent to two single-labeled vertices and thus admits two possible realizations that we use to represent truth assignments.
We refer to these realizations as the \emph{true-state} and the \emph{false-state} of the wedge.

Each clause gadget consists of an edge $ab$ with $L(a) = \{p^1_a, p^2_a\}$ and $L(b) = \{p^1_b, p^2_b\}$, where $p^1_a, p^2_a, p^1_b$, and $p^2_b$ are placed as shown in \Cref{fig:pm3satA}.
The clause gadget additionally contains \emph{literal wedges} corresponding to its three literals, placed in such a way that the edge $ab$ can be realized if and only if at least one of the wedges is in its true-state.
When in true-state, each of the literal wedges extends straight towards its corresponding gadget, along the corresponding edge of the rectilinear representation of the incidence graph.

A variable gadget for a variable $x$ consists of a \emph{variable wedge} $w_x$ representing the truth assignment of $x$.
The true-state of every positive literal wedge that corresponds to $x$ enters the variable gadget from above and intersects the false-state of $w_x$; see \Cref{fig:pm3satA}.
This guarantees that, if $w_x$ is in its false-state, the same holds for the corresponding positive literal wedges.
If $w_x$ is in true-state, these wedges can be in either state.
Analogously, the true-state of the negative literal wedges enter from below and intersect the true-state of $w_x$.

We can show that the resulting instance $(G, S, L)$ of \ListPointEmbedding is a yes-instance if and only if $\phi$ has a satisfying assignment.
We can then turn the graph into a path by connecting the wedges using additional vertices and edges as illustrated in \Cref{fig:pm3satB}.
While most of the new vertices are single-labeled, each clause gadget requires an additional bi-labeled vertex $v$ that ensures that the new edges of the path can be realized, regardless of the chosen realization for the edge $ab$ of the clause gadget. By connecting the two endpoints of the path with an additional path below (dashed line in \Cref{fig:pm3satB}), we can also turn the graph into a cycle.
\end{proofsketch}
\begin{prooflater}{ptheoremHardnessPath}
Starting with the rectilinear representation of the incidence graph of $\phi$, we construct an equivalent instance $\InstanceLong$ of \ListPointEmbedding by replacing variables and clauses with gadgets as follows.
We first present the general approach and subsequently show how the graph of the resulting instance can be augmented to form a path (respectively a cycle).

The basic building blocks of our reduction are called \emph{wedges}. 
A wedge consists of a bi-labeled degree-2 vertex that is adjacent to two single-labeled vertices and thus admits two possible realizations that we use to represent truth assignments.
We refer to these realizations as the \emph{true-state} and the \emph{false-state} of the wedge.

Each clause gadget consists of an edge $ab$ with $L(a) = \{p^1_a, p^2_a\}$ and $L(b) = \{p^1_b, p^2_b\}$, where $p^1_a, p^2_a, p^1_b$, and $p^2_b$ are placed as shown in \Cref{fig:pm3satA}.
Let $x'$, $y'$, and $z'$ denote the three literals of the corresponding clause.
The gadget contains one wedge corresponding to $x'$ whose false-state only intersects the realization $p^1_ap^2_b$ of $ab$, one wedge corresponding to $z'$ whose false-state only intersects $p^2_ap^2_b$, and one wedge corresponding to $y'$ whose false-state only intersects $p^1_ap^1_b$ and $p^2_ap^1_b$.
We refer to these wedges as the \emph{literal wedges}.
In their respective true-states, the wedges intersect none of the possible realizations of the edge $ab$ and instead extend straight towards their respective variable gadget, along the corresponding edge of the rectilinear representation of the incidence graph.
Therefore, the edge $ab$ can be realized if and only if at least one of the three wedges is in its true-state.

A variable gadget for a variable $x$ consists of a \emph{variable wedge} $w_x$ that represents the truth assignment of $x$.
The true-state of every positive literal wedge that corresponds to $x$ enters the variable gadget from above and intersects the false-state of $w_x$; see \Cref{fig:pm3satA}.
This guarantees that, if $w_x$ is in false-state, the same holds for the corresponding positive literal wedges.
If $w_x$ is in true-state, these wedges can be in either state.
Analogously, the true-state of every negative literal wedge that corresponds to $x$ enters from below and intersects the true-state of $w_x$.

\begin{claim}
    \label{lem:pm3sat-reduction}
  $\InstanceLong$ is a yes-instance of \ListPointEmbedding if and only if $\phi$ has a satisfying assignment.
\end{claim}
\begin{claimproof}
  First assume that $\phi$ has a satisfying assignment $\varphi$.
  For each variable $x$ of $\phi$, we embed~$w_x$ and the literal wedges corresponding to $x$ in their true-state or false-state depending on $\varphi(x)$.
  By construction, we obtain no crossings between wedges.
  Moreover, since $\varphi$ is satisfying, at least one literal wedge of each clause gadget is in its true-state.
  As argued above, this implies that each clause gadget can be realized without crossings.
  Hence we obtain a crossing-free realization of $\InstanceLong$.
  
  Conversely, assume $\InstanceLong$ admits a geometric point set embedding $\Gamma$.
  We construct a truth assignment $\varphi$ of $\phi$ as follows.
  For each variable $x$ of $\phi$, we set $\varphi(x) = \top$ if and only if the corresponding wedge $w_x$ is in its true-state.
  It remains to show that $\varphi$ satisfies $\phi$.
  
  Recall that our construction ensures that, if a variable wedge $w_x$ is in its false-state, all corresponding positive literals are in false-state and if $w_x$ is in true-state, all corresponding negative literals are in false-state.
  However, if $w_x$ is in its true-state, the corresponding positive literals can be in either state and if it is in false-state, the corresponding negative literals can be in either state.
  For each clause~$c$, the assignment $\varphi'$ that the literal wedges of~$c$ in $\Gamma$ induce for the literals of $c$ therefore does not necessarily coincide with $\varphi$.
  However,~$\varphi'$ can be obtained from $\varphi$ by flipping the truth value of some literals from $\top$ to $\bot$.
  Since the clause gadget only admits a realization if at least one of its literal wedges is in true-state, the assignment $\varphi'$ (and thus also $\varphi$) satisfies $c$.
  Therefore, $\varphi$ satisfies $\phi$.
\end{claimproof}

We now show how the resulting instance can be connected to form a path.
\Cref{fig:pm3satB} shows the route we want the path to take through the incidence graph.
The path starts to the left of the first variable, visits all positive literal wedges of the variable from left to right, then uses the variable wedge to return to the left and symmetrically proceeds for the negative literal wedges; see \Cref{fig:pm3satB}.
Subsequently, the path exits the variable gadget to the right and proceeds to the next variable.

The second literal wedge of each clause additionally has the special purpose of connecting to the isolated edge $ab$ of the clause gadget.
To this end, the path enters the clause gadget, visits the edge $ab$, and subsequently connects to the left endpoint $u$ of the literal wedge; see \Cref{fig:pm3satB}.
To ensure that the connection from $b$ to $u$ can always be realized, we add an additional bi-labeled vertex $v$ adjacent to $b$ and $u$ whose corresponding two points are placed as shown in \Cref{fig:pm3satB}.
This way, the additional edges of the path can always be realized, regardless of the realization of the edge $ab$.

Note that we can always add additional vertices and edges to connect the two endpoints of the path and thus turn the graph into a cycle.
Since the reduction can be computed in polynomial, time this concludes the proof.
\end{prooflater}

\subsection{A Fixed-Parameter Algorithm by the Vertex Cover Number}
\label{sec:general-vc}
In this section, we explore the complexity of \ListPointEmbedding when parameterized by the vertex cover number $\vcn(G)$ of a bi-labeled graph $G$.
The vertex cover number denotes the minimum number of vertices whose deletion removes all edges.
Let $\Instance = \InstanceLong$ be a bi-labeled instance of \ListPointEmbedding and let $C \subseteq V$ be a vertex cover of $G$ of size $k = \Size{C}$.

We first branch over all realizations of the vertices from $C$ on points in $S$.
Since \Instance is bi-labeled, we have $\Size{L(p)} \leq 2$ for all $p \in C$.
Consequently, this yields at most $2^k$ different branches.
In each branch, the placement of vertex cover vertices on $S$ is fixed and it only remains to check whether this partial realization $\Drawing'$ can be extended to one of $G$ by placing the remaining vertices $V(G) \setminus C$.
Observe that for each $u \in V(G) \setminus C$, all neighbors are in~$C$.
Hence, the drawing of all edges $\{u,v\} \in E(G)$ is completely determined once we have placed~$u$.
The fact that $u$ is bi-labeled allows us devise a \probname{2SAT}-formula to check if there is a realization $\Drawing$ such that $ \Drawing(v) = \Drawing'(v)$, for all $v \in C$\onlyShort{; see %
\NewText{the full version~\cite{ARXIV}} for details}.

\begin{statelater}{vertexCoverSat}
\begin{restatable}\restateref{lem:sec:parameterized-algorithms-vc-sat}{lemma}{lemmaVCSat}
	\label{lem:sec:parameterized-algorithms-vc-sat}
	For a bi-labeled instance $\InstanceLong$, let $\Drawing'$ be a realization of a vertex cover $C$ of $G$ of size $k$.
	We can check in $\BigO{n^2 \cdot (s + n\log n)}$ time if there exists a realization $\Drawing$ that extends $\Drawing'$.
\end{restatable}
\begin{proof}
	Let $\Instance = \InstanceLong$, $C$, and $\Drawing'$ be as in the statement.
	Given $\Drawing'$, it only remains to place the non-vertex cover vertices $U \coloneqq V(G) \setminus C$ of $G$.
	We now describe how to find such a placement in linear time with the help of a \probname{2Sat} formula $\phi$.

	\proofsubparagraph*{Variables.}
	Let $v \in U$.
	As $\Instance$ is bi-labeled, we have $\Size{L(v)} \leq 2$.
	We can assume without loss of generality $L(v) = \{p, q\}$ since the placement of $v$ is otherwise fixed.
    In the following, we refer for a vertex $v \in U$ with $p(v)$ and $q(v)$ to the two points in $L(v)$.
	For $v \in U$, we introduce two variables $x_{v,p(v)}$ and $x_{v,q(v)}$ that indicate with a truth assignment $\Upsilon$ whether $\Gamma(v) = p$ ($\Upsilon(x_{v,p}) = 1$) or $\Gamma(v) = q$ ($\Upsilon(x_{v,q}) = 1$) should hold.\footnote{We remark that in principle one variable per vertex would suffice. However, having two variables facilitates the presentation of the formula. Moreover, note that in both cases the number of variables is asymptotically linear in the number of vertices.}
	In total, this gives us $2\cdot\Size{U} \in \BigO{n}$ variables.
		
	\proofsubparagraph*{Clauses.}
	Our formula $\phi$ consists of four subformulas, $\phi_1$ to $\phi_4$, each targeted at ensuring a different property of a potential solution.
	
	The first two subformular, $\phi_1$ and $\phi_2$, ensure that each vertex $v \in U$ is placed on exactly one point from $L(v) = \{p,q\}$.
	Subformular $\phi_1$ ensures that each vertex is placed on at least one point:
	\begin{align*}
		\phi_1 \coloneqq \bigwedge_{v \in U} (x_{v, p(v)} \lor x_{v, q(v)})
	\end{align*}
	Similarly, $\phi_2$ ensures that each vertex is placed on at most one point:
	\begin{align*}
		\phi_2 \coloneqq \bigwedge_{v \in U} (\lnot x_{v, p(v)} \lor \lnot x_{v, q(v)})
	\end{align*}
	
	With $\phi_3$, we ensure that we assign to every point $p \in S$ at most one vertex.
	On the one hand, we have to ensure that no vertex $v \in U$ is assigned to a point $p \in S$ with $\Drawing'(u) = p$ for some $u \in C$.
	On the other hand, no two vertices $u,v\in U$ can be assigned to $p$ simultaneously.
	For the former constraint, we define $S' \coloneqq \{\Drawing'(v) \mid v \in C\}$.
	For the latter constraint, let $K \coloneqq \{\{u,v\} \in \binom{U}{2} \mid L(u) \cap L(v) \neq \emptyset\}$ denote the vertices that share at least one admissible point.
	With $S'$ and $K$ at hand, we can define $\phi_3$ as follows:
	\begin{align*}
		\phi_3 \coloneqq \left(\bigwedge_{p \in S'} \bigwedge_{v \in \{v \in U \mid p \in L(v)\}} \lnot x_{v,p}\right) \land \left(\bigwedge_{\{u,v\} \in K} \bigwedge_{p \in L(u) \cap L(v)} (\lnot x_{u, p} \lor \lnot x_{v, p})\right)
	\end{align*}
	
	Finally, we use $\phi_4$ to encode that the solution should be crossing-free.
	Recall that for each $v \in U$ and all edges $\{u,v\} \in E(G)$, we have $u \in C$.
	Hence, when placing a vertex $v \in V$ on a point $p \in L(v)$, the resulting drawing $\Drawing' \cup (v \to p)$ contains all of its incident edges $\{u,v\} \in E(G)$.
	Thus, it only remains to ensure that for every pair of vertices $u,v \in U$, the resulting realization is planar, i.e., forbid partial realizations that contain crossings.
	For $u, v \in \binom{U}{2}$, we let $X(u,v) \coloneqq \{(p, q) \mid p \in L(u), q \in L(v) \setminus \{p\}, \Drawing' \cup \{(u \to p), (v \to q)\}\ \text{contains a crossing}\}$ denote all admissible assignments of $u$ and $v$ that cross.
	Using $X$, we can now define $\phi_4$ as follows:
	\begin{align*}
		\phi_4 \coloneqq \bigwedge_{u,v \in \binom{U}{2}} \bigwedge_{(p,q) \in X(u,v)} (\lnot x_{u,p} \lor \lnot x_{v,q})
	\end{align*}
	
	\proofsubparagraph*{Putting all together.}
	This completes the definition of our \probname{2Sat} formula $\phi = \phi_1 \land \phi_2 \land \phi_3 \land \phi_4$.
	A close analysis of $\phi$ reveals that it can be constructed in $\BigO{n^2 \cdot (s + n\log n)}$ time.
    For~$\varphi_4$, observe that for two fixed vertices $u, v \in V$, there are at most four different placements (since~$\Instance$ is bi-labeled).
    For each such placement, we can check in $\BigO{s}$ time if it is admissible and use a modified version of the Bentley–Ottmann segment intersection test~\cite{dBCvKO.CGA.2008} to check in $\BigO{n \log n}$ time if two edges cross in a given drawing.
    Moreover, there are $\BigO{n^2}$ choices for~$u$ and $v$.
	Since $\phi$ has $\BigO{n}$ variables (and polynomially-many clauses), we can check in linear time if it has a satisfying assignment.
	It remains to argue correctness of the formula.
	
	For the forward-direction ($\Rightarrow$), we take a satisfying assignment $\Upsilon$ of $\phi$ and construct a realization $\Drawing$ as follows.
	Initialize $\Drawing$ with $\Drawing'$ and set $\Drawing(v) = p$ if and only if $\Upsilon(x_{v, p}) = 1$.
	Subformulas $\phi_1$ and $\phi_2$ ensure that in $\Drawing$ every vertex is assigned to exactly one of its admissible points.
	Together with $\phi_3$, they furthermore ensure that $\Drawing$ is well-defined, in particular, no two vertices are assigned to the same point.
	Finally, assume that two edges $e_1 = \{u_1,v_1\}$ and $e_2 = \{u_2, v_2\}$ cross in $\Drawing$.
	We can identify at least one clause in $\phi_4$ that is not satisfied by $\Upsilon$:
	Without loss of generality $u_1 \in U$ holds.
	If $e_2 \cap C = e_2$, then any clause with the variable $x_{u_1, \Drawing(u_1)}$ is not satisfied as $e_1$ crosses an edge between two vertex cover vertices.
	Otherwise, without loss of generality, assume $u_2 \in U$.
	As $e_1$ and $e_2$ cross in $\Drawing$, they also cross in $\Drawing' \cup \{(u_2 \to \Drawing(u_1)), (u_2 \to \Drawing(u_2))\}$.
	Hence, by the definition of $\Drawing$, the clause $x_{u_1, \Drawing(u_1)} \lor x_{u_2, \Drawing(u_2)}$ is not satisfied.
	Consequently, $e_1$ and $e_2$ can not cross and altogether $\Drawing$ is a realization that extends $\Drawing'$.	
	
	For the backward-direction ($\Leftarrow$), take a realization $\Drawing$ that extends $\Drawing'$.	
	We construct a truth assignment $\Upsilon$ by setting $\Upsilon(x_{v, p})$ if and only if $\Drawing(v) = p$, for all $v \in U$.
	As $\Drawing$ is a realization, this is well-defined.
	To see that $\Upsilon$ satisfies $\phi$, it suffices to observe that it satisfies~$\phi_1$ and $\phi_2$ since every vertex, and in particular every $v \in U$, is placed one of its admissible points.
	Furthermore, as $\Drawing$ is a realization, no two vertices are assigned on the same point and no two edges cross.
	Hence, $\Upsilon$ satisfies $\phi_3$ and $\phi_4$, respectively.
	Since it satisfies all subformulas, it also satisfies $\phi$.
	Combining all, the lemma statement follows.
\end{proof}
\end{statelater}

\onlyLong{
We now use \Cref{lem:sec:parameterized-algorithms-vc-sat} to show \Cref{thm:parameterized-algorithms-vc-bilabeled}.
}
\begin{restatable}\restateref{thm:parameterized-algorithms-vc-bilabeled}{theorem}{theoremVCBiLabeled}
	\label{thm:parameterized-algorithms-vc-bilabeled}
	\ListPointEmbedding for a bi-labeled graph $G$ can be solved constructively in time $\BigO{2^{\vcn(G)}\cdot (n^2 \cdot (s + n \log n)}$ and is thereby \FPT\ in $\vcn(G)$.
\end{restatable}
\begin{prooflater}{ptheoremVCBiLabeled}
	Let $\Instance = \InstanceLong$ be a bi-labeled instance of \ListPointEmbedding.
	We first compute a minimum vertex cover $C \subseteq V$ of $G$.
	Such a vertex cover of size $\Size{C} = \vcn(G) = k$ can be computed in $\BigO{2^{k} + k\cdot n}$ time~\cite{CKX.Iub.2010}.
	We then branch to determine the realization of $C$.
	As $\Instance$ is bi-labeled, this yields $\BigO{2^k}$ different branches.
	Furthermore, as the branching is exhaustive, we preserve the existence of a solution.
	In each branch, we first check if the drawing $\Drawing'$ is well-defined and planar, i.e., if no two vertices are assigned to the same point and if no two edges cross.
	If this is not the case, we discard the branch.
	Overall, this takes $\BigO{s + k \log k}$ time since we can use a Bentley–Ottmann segment intersection test to check for crossings~\cite{dBCvKO.CGA.2008}.
	If $\Drawing'$ is a partial realization, we use \Cref{lem:sec:parameterized-algorithms-vc-sat} to check in if the branch can be extended to a solution for \Instance.
	Overall, this takes $\BigO{s + k \log k + (n^2 \cdot (s + n \log n)}$ time per branch.
	Combining this with the preprocessing time to compute $C$ and using $k \leq n$, the statement follows.
\end{prooflater}

\subsection{Realizing Two Tri-Labeled Stars is Hard}
\label{sec:tri-labeled-vertex-cover-2}

While \ListPointEmbedding of bi-labeled graphs is \FPT\ with regard to the vertex cover number, this changes if vertices have more than two labels. The previous section is tight in the sense that \ListPointEmbedding is hard for tri-labeled graphs already for graphs with vertex cover number two.

\begin{figure}
		\centering
		\includegraphics[scale=0.7]{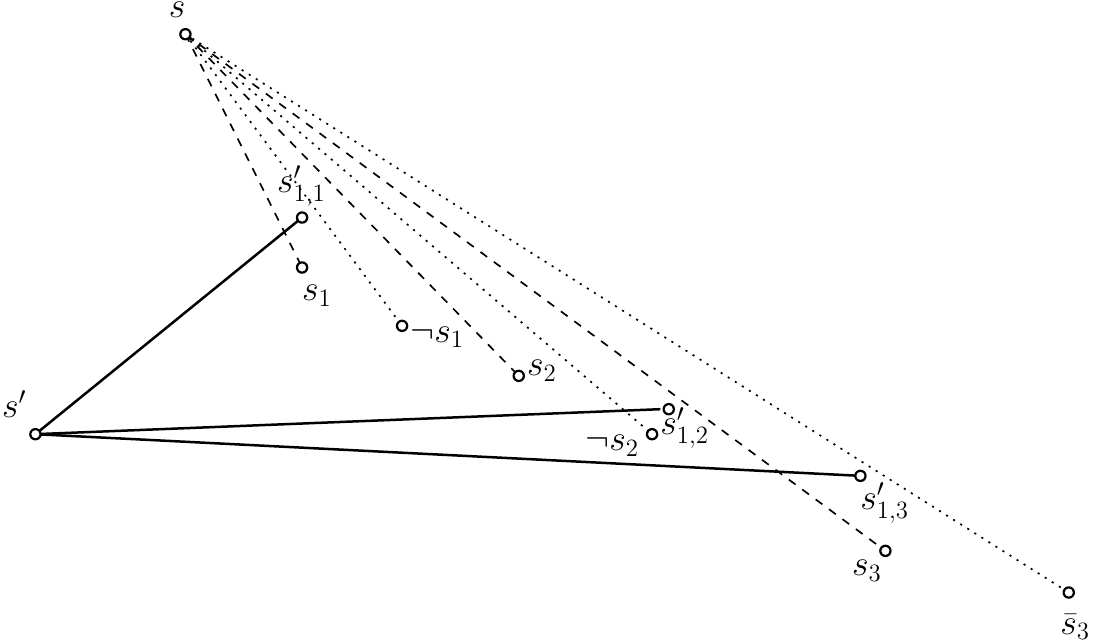}
		\caption{Reduction representing the clause~$\lnot{x}_1\lor x_2 \lor \lnot{x}_3$.}
		\label{fig:vertexcoverinstance-appendix}
\end{figure}

We reduce from \textsc{3SAT}. For this purpose assume, we have a \textsc{3SAT}-formula $\phi$ on variables $x_1,\ldots,x_n$ and clauses $C_1,\ldots,C_m$ in conjunctive normal form.
We represent all variables and all clauses by one star, respectively.

\subparagraph{Star of Variables.} We add one vertex $v$ to $V(G)$ and a point $s$ to $S$ such that $L(v)=\{s\}$. For every variable $x_i$ we add a vertex $v_i$ and two points $s_i$ and $\lnot s_i$ and set $L(v_i)=\{s_i,\lnot s_i\}$. Further, we add to $E(G)$ the edge $\{v,v_i\}$ such that in the embedding the straight line drawing of the graph will either contain the edge $\{s,s_i\}$ or the edge $\{s,\lnot{s}_i\}$. The points $s_i$ and $\lnot{s}_i$ for all variables together will be spread around $s$ such that the resulting edges form a fan around $s$ as seen in Figure~\ref{fig:vertexcoverinstance-appendix}.

\subparagraph{Star of Clauses.} We add one vertex $v'$ to $V(G)$ and a point $s'$ to $S$ such that there exists a straight line from $s'$ to any of the vertices $s_i$ or $\lnot{s}_i$ that is not crossed by any other edge of that form. Further $L(v') = \{s'\}$. For every clause $C_j$ we add a vertex $v'_j$ to $V(G)$ and an edge $\{v',v'_j\}$ to $E(G)$. Further, we add three points $s'_{j,1}$, $s'_{j,2}$ and $s'_{j,3}$ to $S$ one for each literal in $C_i=l_{j,1}\lor l_{j,2} \lor l_{j,3}$ such that if $l_{j,1} = x_i$, then the edge $\{s',s_{j,1}\}$ crosses the edge $\{s_{i,1},s_{i,2}\}$ and if $l_{j,1} = \lnot{x}_i$, then the edge $\{s',s_{j,1}\}$ crosses the edge $\{s_{i,1},\lnot{s}_{i,2}\}$. We repeat for $l_{j,2}$ and $l_{j,3}$. We conclude by setting $L(v'_j) = \{s'_{j,1},s'_{j,2},s'_{j,2}\}$.

\begin{restatable}\restateref{thm:threecover}{lemma}{threecover}
		\label{thm:threecover}
		$G=(V(G),E(G))$ can be embedded on $S$ while respecting $L$ if and only if~$\phi$ has a satisfying assignment.
\end{restatable}
\begin{proof}%
    First, assume $\phi$ has a satisfying assignment. For the vertices $v$ and $v'$ there is only one option to place the vertex, so we place the vertices on that point. For every variable $x_i$, if $x_i$ has the value \textsc{True}, then we place the vertex~$v_{i}$ on $\lnot{s}_{i}$. By this, we also add the edge $\{s_{i},\lnot{s}_{i}\}$ to the drawing. Otherwise, if $x_i$ has as the value \textsc{False}, then we place the vertex~$v_{i}$ on $s_{i}$, adding the edge $\{s,s_{i}\}$ to the drawing. For every clause $C_j$, pick one of the literals $l_{j,1},l_{j,2}$ or $l_{j,3}$ in $c_j$ that is satisfied and place the vertex $v'_j$ on $s'_{j,1}, s'_{j,2}$ or $s'_{j,3}$ respectively. If, without loss of generality, $l_{j,1}$ is satisfied and has value \textsc{True}, then the resulting edge $\{s',s'_{j,1}\}$ will not be crossed. The only edge that could cross $\{s',s'_{j,1}\}$ is $\{s,s_{i}\}$ such that $l_{j,1}$ is the positive form of $x_i$, but this edge does not exist, since we added $\{s_{i,1},\lnot{s}_{i,2}\}$ instead, by the choices made based on the truth assignment. We conclude that all vertices are placed and the resulting straight line embedding of the graph is plane.

    Conversely, assume that $G=(V(G),E(G))$ can be embedded on $S$ while respecting~$L$. We go through all the variable gadgets and if the vertex $v_i$ is placed on point $s_{i}$, we assign to $x_i$ the value \textsc{False}, conversely if $v_i$ is placed on $\lnot{s}_{i}$ we assign \textsc{True}. The resulting assignment is a satisfying assignment for $\phi$ since in every gadget corresponding to a clause~$c_j$, the vertex has been placed, without loss of generality, on point $s'_{j,1}$. This implies that the edge $\{s',s'_{j,1}\}$ is not crossed in the drawing and, by construction, the literal $l_{j,1}$ is satisfied. Therefore, every clause has at least one satisfied literal and consequently $\phi$ has a satisfying assignment.
\end{proof}

\begin{figure}[t]
		\centering
		\includegraphics[scale=0.7,page=2]{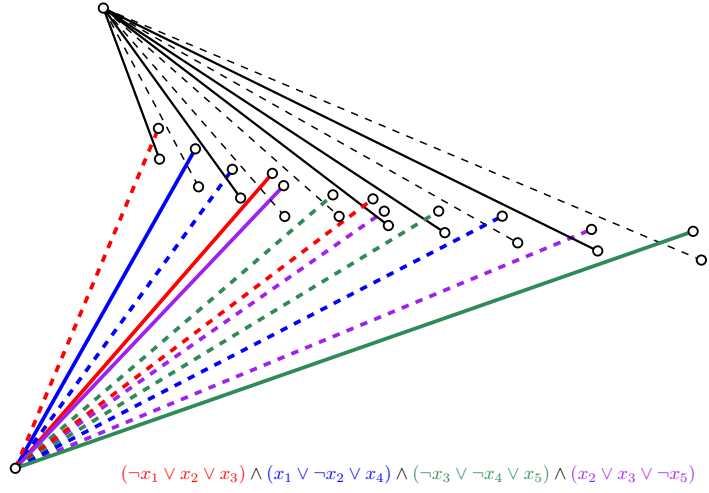}
		\caption{%
        The black edges of the top star correspond from left to right to the truth assignments $x_1=\textsc{True}$, %
        $x_2 = \textsc{True}$, $x_3=\textsc{False}$, $x_4=\textsc{True}$\NewText{, and} $x_5 = \textsc{True}$. The edges of the bottom star correspond to the three possible ways to realize an edge of the clause in the same color.} 
		\label{fig:vertexcoverinstance}
\end{figure}
We provide in \cref{fig:vertexcoverinstance} an illustration of the reduction and conclude:

\begin{restatable}\restateref{thm:threecoverx}{theorem}{threecoverx}
		\label{thm:threecoverx}
		\ListPointEmbedding is \NP-complete even for tri-labeled graphs $G$ %
        with $\vcn(G) = 2$.
\end{restatable}

\subparagraph{Bi-labeled Optimization Version.}
We contrast the positive results for \ListPointEmbedding from a paramatrized perspective with negative results the following optimization variant. Given a bi-labeled graph $G=(V,E)$ (already with vertex cover number $2$) and a point set~$S$, embed the largest possible subgraph of $G$ (with respect to the number of edges) on $S$. %

In particular, we show that this optimization variant of our problem is \APX-hard by modifying our \NP-hardness reduction for the tri-labeled case to a \PTAS-reduction from MAX-2SAT(3)~\cite{vazirani2001approximation}. An instance of MAX-2SAT(3) has a a 2SAT formula $\phi$ as input in which every variable appears in at most 3 clauses. The goal is to satisfy as many clauses as possible by choosing an according truth assignment.

\subparagraph{Star of Variables.} Instead of one vertex $v_i$ with $L(v_i)=\{x_i,\lnot s_i\}$ to represent the assignment of every variable, we make four copies $v_i$, $v_i^{1}$ to $v_i^{4}$ with $L(v_i^j) = \{s_i^j,\lnot s_i^j\}$. If previously straight lines corresponding to clauses crossed the line $ss_i$, then now they cross all four lines $ss_i^j$. A symmetric statement can be made about the points representing the negated variable.

\subparagraph{Star of Clauses.} Observe that now every vertex corresponding to a clause is bi-labeled instead of tri-labeled, since every clause only contains two variables instead of three.

Observe that we can always embed at least $4n$ edges without crossings, namely the edges that correspond to the star of variables.

\begin{restatable}\restateref{lem:preapx}{lemma}{lemmapreapx}
\label{lem:preapx}
    We can embed $4N+k$ edges of $G$ on $S$ while respecting $L$ if and only if $\phi$ has an assignment satisfying $k$ clauses.
\end{restatable}
\begin{proof}%
    If there is an assignment satisfying $k$ clauses, then we embed the edges corresponding to the variable star as explained in \cref{thm:threecover}. For each satisfied clause, we can embed one additional edge, hence, in total $4N+k$ edges are embedded.

    If $4N+k$ edges can be embedded, then there exists an emedding with the same size in which the entire star of variables and $k$ edges of the star of clauses are embedded: First, not embedding all edges of a variable in order to embed all edges corresponding to clauses containing this variable results in fewer embedded edges, since a variable occurs in at most three clauses. Second, placing all points for a variable on $s_i^j$ or all points on $\lnot s_i^j$ is never worse than distributing between them.
    
    Now, we can deduce an assignment satisfying $k$ clauses the same way as in the proof of \cref{thm:threecover}.
\end{proof}

\begin{restatable}\restateref{thm:preapx}{theorem}{theorempreapx}
\label{thm:preapx}
    Finding a realization with the maximum number of edges for a bi-labeled graph $G$ is \APX-hard even if $\vcn(G) = 2$.

\end{restatable}

\begin{proof}%
    We show that this problem is \APX-hard by modifying our \NP-hardness reduction for the tri-labeled case to a \PTAS-reduction from MAX-2SAT(3)~\cite{vazirani2001approximation}. An instance of MAX-2SAT(3) has a a 2SAT formula $\phi$ as input in which every variable appears in at most 3 clauses. The goal is to satisfy as many clauses as possible by choosing an according truth assignment. Instead of one single variable-edge, we use four edges that all cross the clause-edges that represent conflicting truth assignments. Clause-edges only have two possible embedding, instead of three, since clauses only contain two literals.

    Let $k^\ast$ be the maximum number of clauses, we can satisfy in a 2SAT(3) formula $\phi$. Note that $k^\ast > \frac{3}{4}M$, since a random assignment of the variables satisfies $\frac{3}{4}$ of the clauses and hence there exists an assigment that satisfies at least this many clauses. Further, observe that we can omit variables from $\phi$ that do not appear in any clause. Additionally, for variables that appear in only a single clause, we trivially assign the truth value that satisfies said clause. Thus, we omit the variable and the clause. So we can without loss of generality assume that all the remaining variables appear in at least two clauses. Since every clause contains two variables, we get $M\geq N$.

    Now assume, we could approximate the optimal number of embedded edges $e^\ast$ by a factor of $(1-\epsilon)$. By \cref{lem:preapx} we get that $e^\ast = 4N+ k^\ast$. By subtracting $4N$ we can calculate that we can approximate $k^\ast$:

    \begin{align*}
        (1-\epsilon)(e^\ast) - 4N & = (1-\epsilon)(4N+k^\ast) - 4N = (1-\epsilon)k^\ast - 4\epsilon N  \\ & \geq (1-\epsilon)k^\ast - 4 \epsilon M \geq (1-\epsilon)k^\ast - 3 \epsilon k^\ast = (1-4\epsilon)k^\ast
    \end{align*}

    Now if there were a \PTAS, approximating the number of embedded edges then there would also be one for the number of satisfied clauses in $\phi$ which is not possible, unless \P=\NP.
\end{proof}

\subsection{Efficient Algorithms for Two Stars in Convex Position or One Star}
\label{sec:parameterized-algorithms-star}
Recall \Cref{thm:threecoverx}, which essentially rules out any general algorithm parameterized by the vertex cover number and thus ``forced us'' to consider bi-labeled instances.
Since \Cref{thm:threecoverx} ``only'' shows hardness for instances with vertex cover number two, it is natural to ask what happens if the vertex cover number is one, i.e., if $G$ is a single star.
In this section, we show that we can solve such instances in polynomial time, independent of the structure of $L$.

The high-level idea of our approach is as follows.
Since $G$ is a star, all edges are incident to its center $v^* \in V(G)$, i.e., no two edges can cross.
Thus, we only need to ensure that no two vertices are placed on the same point.
To this end, we reduce the problem of finding a realization to the problem of finding a matching of a certain size in a carefully constructed bipartite auxiliary graph, which can be done in polynomial time.

\begin{restatable}\restateref{thm:parameterized-algorithms-star}{theorem}{theoremStar}
	\label{thm:parameterized-algorithms-star}
    \NewText{If we can find a maximum matching in an $\eta$-vertex $\theta$-edge bipartite graph in time $\BigO{f(\eta, \theta)}$, then
    \begin{itemize}
        \item \ListPointEmbedding, where $G$ is a single star, can be solved constructively in time $\BigO{f(s + n, s \cdot n) + s \cdot n}$ and
        \item \ConvexListPointEmbedding, where $G$ is the disjoint union of two stars, can be solved constructively in time $\BigO{s^2 \cdot (f(s+n, s \cdot n) + s \cdot n)}$.
    \end{itemize}    
    }
\end{restatable}
\NewText{For example, we can apply in \Cref{thm:parameterized-algorithms-star} the $\BigO{\sqrt{\eta}\cdot\theta}$-time algorithm from Hopcroft and Karp~\cite{HK.AMM.1973}.
Plugging the running time into \Cref{thm:parameterized-algorithms-star} (and observing $s + n \leq 2s$) yields a $\BigO{s^{1.5} \cdot n}$-time constructive algorithm for \ListPointEmbedding where $G$ is a single star.
Moreover, it yields a $\BigO{s^{3.5} \cdot n}$-time constructive algorithm for \ConvexListPointEmbedding where $G$ is the disjoint union of two stars.}
\begin{proofsketch}[Proof sketch of \Cref{thm:parameterized-algorithms-star}.]
	We construct an auxiliary bipartite graph $H$ that captures with its edges the admissible placements of all vertices $v \in V(G)$ onto points $p \in S$.
    More concretely, we have $V(H) = V(G) \partition S$ and $\{v,p\} \in E(H)$ if and only if $p \in L(v)$.
    \NewText{Observe that $H$ has $n + s$ vertices and at most $s \cdot n$ edges.
    It can be constructed in time $\BigO{s \cdot n}$.}
    Every matching $F \subseteq E(H)$ corresponds to a realization $\Drawing'$ of the subgraph $G' \subseteq G$ induced %
    \NewText{by} the endpoints of $F$ that correspond to the vertices in $V(G)$.
	We now can %
    \NewText{compute} a maximum matching $F^* \subseteq E(H)$ in $H$ and check if $\Size{F^*} = \Size{V(G)} = n$. 

    If $G$ consists of two disjoint stars and $S$ is in convex position, we observe that every solution $\Drawing$ on a convex point set $S$ has a separating line between the two realizations; see also \Cref{fig:two-stars-convex}.
    We can enumerate all such lines to obtain the statement.
\end{proofsketch}
\begin{prooflater}{ptheoremStar}[Proof of \Cref{thm:parameterized-algorithms-star}.]
	Let $\Instance = \InstanceLong$ be an instance of \ListPointEmbedding where $G$ is a star.
	We construct an auxiliary bipartite graph $H$ that captures with its edges the admissible placements of all vertices $v \in V(G)$ onto points $p \in S$, i.e., $\{v,p\} \in E(H)$ if and only if $p \in L(v)$.
	Formally, $V(H) = V(G) \partition S$ and $E(H) = \{\{v,p\} \mid v \in V(G), p \in S \cap L(v)\}$.
	Note that $H$ has $n + s$ vertices, as $n \leq s$ must hold, and at most $s \cdot n$ edges.
	It can be constructed in $\BigO{s \cdot n}$ time.
	
	We now make the following crucial observation.
	Every matching $F \subseteq E(H)$ corresponds to a realization $\Drawing'$ of the subgraph $G' \subseteq G$ induced on the endpoints of $F$ that correspond to the vertices in $V(G)$.
	In particular, as $F$ is a matching, every vertex $v \in V(G')$ is placed onto exactly one point $p \in S$ and no two vertices can be placed at the same point.
	By the definition of $H$, every vertex is placed at an admissible point.
	Moreover, we observe that no two edges can cross, as all edges are incident to the center of the star.
	\NewText{Let $F^* \subseteq E(H)$ be a maximum matching in $H$.}
    By above arguments, \Instance admits a solution $\Drawing$ if and only if $\Size{F^*} = \Size{V(G)} = n$.
	Combining all, the statement follows.

    If $G$ consists of two disjoint stars and $S$ is in convex position, we observe that every solution $\Drawing$ on a convex point set $S$ has a separating line between the two realizations; see also \Cref{fig:two-stars-convex}.
    We can enumerate all such lines \NewText{and treat it as two separate instances}.
\end{prooflater}

\begin{figure}
	\centering
	\includegraphics{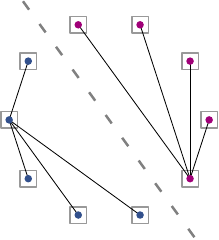}
	\caption{The dashed line separates the blue and purple star in $\Drawing$ on a convex point set.}
	\label{fig:two-stars-convex}
\end{figure}

\section{Extending Partial Realizations}
\label{sec:extension}
\Cref{sec:convex,sec:general} demonstrate the (computational) difficulty of \ListPointEmbedding.
In recent years, extending incomplete drawings has proven to be a successful perspective onto difficult problems~\cite{eghkn-ep1d-20,dfgn-pcesl-24,DFGN.PCE.2026,bgkmn-eopgdft-23,bgkmn-eopgdft-25,PLP,QLE} and in this section, we attack \ListPointEmbedding from this angle.
A partial realization~$\Gamma_H$ for some subgraph $H \subseteq G$ can be encoded by setting $L(v) = \{\Gamma_H(v)\}$ for every $v \in V(H)$.
In the following, we, therefore, define $\Vadd = \{v \in V(G) \mid \Size{L(v)} > 1\}$ to be the remaining \emph{missing} vertices and let $\nadd = \Size{\Vadd}$ denote their number.
We call the corresponding problem \PartialListPointEmbedding and analyze its parameterized complexity with respect~to~\nadd.

\shortLong{
\subparagraph*{\textsf{XP}-Membership.}
First, we show that \PartialListPointEmbedding is in \XP ~with respect to \nadd.
To this end, observe that we can branch on all possible placements of the \nadd missing vertices on their $\BigO{s}$ admissible points.
For each branch, we can determine in $\BigO{s + n \log n}$ time if it corresponds to a solution using the Bentley–Ottmann segment intersection test~\cite{dBCvKO.CGA.2008}.
}{
\subsection{\textsc{PGLPSE} Parameterized by the Number of Missing Vertices is in \textsf{XP}}
\label{sec:extension-xp}
As our first result under the parameterization of the complex vertices, we show that \PartialListPointEmbedding is \XP-tractable with respect to \nadd.
This ``warm-up'' result can be achieved via a text-book branching strategy.
}

\begin{restatable}\restateref{thm:extension-xp}{theorem}{theoremExtensionXP}
	\label{thm:extension-xp}
	\PartialListPointEmbedding can be solved constructively in time $\BigO{s^{\nadd + 1}\cdot n\log n}$ and is thereby in \XP\ when parameterized by \nadd.
\end{restatable}
\begin{prooflater}{ptheoremExtensionXP}
	Let $\Instance = \InstanceLong$ be an instance of \PartialListPointEmbedding.
	Furthermore, let $U \coloneqq V(G) \setminus \Vadd$ denote the single-labeled vertices.
	We first place all vertices from $u$ on their unique admissible point.
	Afterwards, we exhaustively branch to determine the placement of $\Vadd$.
	Each of the $\BigO{s^{\nadd}}$ branches corresponds to a drawing $\Drawing$ in which each vertex is placed on an admissible point by construction.
	We use the Bentley–Ottmann segment intersection test~\cite{dBCvKO.CGA.2008} to check in $\BigO{s + n\log n}$ time if $\Drawing$ is a solution.
	Finally, observe that our exhaustive branching preserves the existence of a solution, i.e., \Instance has a solution if and only if there exists at least one branch where $\Drawing$ is a realization.
\end{prooflater}

\shortLong{
\subparagraph*{\textsf{W}[1]-hardness.}
Next, we complement \Cref{thm:extension-xp} and show that \PartialListPointEmbedding\ is \W[1]-hard when parameterized by \nadd.
}{
\subsection{A \textsf{W}[1]-hardness Reduction for the Number of Missing Vertices}
\label{sec:extension-w1}
In this section, we complement \Cref{thm:extension-xp} and show that \PartialListPointEmbedding\ is \W[1]-hard when parameterized by the number \nadd of missing vertices.
}%
To this end, we reduce from the problem \probname{Grid Tiling}, defined as follows.
Given two integers $\eta, \kappa \geq 1$, and a family $\mathcal{X}$ of $\kappa^2$ sets, $X_{i,j} \subseteq [\eta]^2$, $i,j \in [\kappa]$, in the following called \emph{tiles}, is there a tuple $x_{i,j} = (a, b) \in X_{i,j}$ for each $i,j\in[\kappa]$ such that if $x_{i,j} = (a,b)$, then $x_{i+1,j} = (a', b)$, for all $i \in [\kappa - 1], j \in [\kappa]$, and $x_{i,j+1} = (a,b')$ for all $i \in [\kappa]$ and $j \in [\kappa - 1]$, i.e., they agree on the second and first value, respectively; see also \Cref{fig:grid-tiling-example}.
\probname{Grid Tiling} is known to be \W[1]-hard parameterized by~$\kappa$.

Let $(\eta, \kappa, \mathcal{X})$ be an instance of \probname{Grid Tiling}.
We construct an instance \InstanceLong\ of \PartialListPointEmbedding\ with $\kappa^2$ missing vertices and $s, \ell \in \BigO{\eta^2}$.
On a high level, we represent each of the $\kappa^2$ tiles $X_{i,j}$, $i,j\in [\kappa]$, with $\Size{X_{i,j}}$ points, one for every $(a,b) \in X_{i,j}$.
The graph $G$ contains a $\kappa \times \kappa$-grid $K \subseteq G$, and placing a designated vertex $v_{i,j} \in V(K)$ on some point $p \in S$ will be equivalent to selecting the corresponding tuple from $X_{i,j}$.
To ensure that we preserve the existence of a solution, we introduce additional auxiliary vertices and edges to~$G$.
\shortLong{%
We use $L(\cdot)$ to essentially fix their placement on the points from $S$.
These vertices, and in particular their edges, will result as \emph{obstacles} in our instance; see the gray boxes in \Cref{fig:grid-tiling-example}.
Essentially, the obstacles leave only a narrow corridor which enforces that the grid~$K$ is drawn as an axis-aligned $\kappa\times\kappa$-grid in any solution, ensuring that the vertex placements agree on their $x$ and $y$-coordinates, respectively.
We provide in \Cref{fig:grid-tiling-example} a small example of the reduction and refer to %
\NewText{the full version~\cite{ARXIV}} for the details.
}
{
In the following, we give the full details of our construction; see also \Cref{fig:grid-tiling-example} for an example.
}
\begin{figure}
	\centering
	\includegraphics{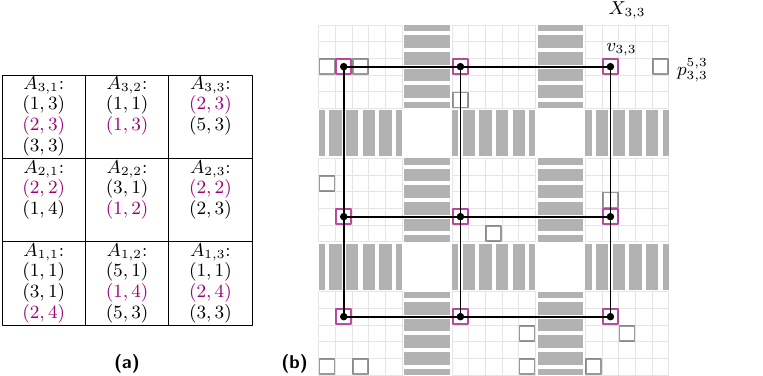}
	\caption{\textbf{\textsf{(a)}} Sample instance of \probname{Grid Tiling} adapted from Cygan et al.~\cite{CFK+.PA.2015} and \textbf{\textsf{(b)}} the corresponding instance of \PartialListPointEmbedding. We indicate a solution in purple. Obstacles are visualized as gray blocks and the light-gray $\eta\times \eta$-grids are for reference only.}
	\label{fig:grid-tiling-example}
\end{figure}
\begin{statelater}{detailsWOne}
\subparagraph*{Encoding the Tiles.}
We first describe how we encode the tiles in the our instance of \ListPointEmbedding.
Let $X_{i,j} \in \mathcal{X}$, $i,j \in [\kappa]$ be a tile of $\mathcal{X}$.
We introduce in $S$ a point $p_{i,j}^{a,b}$ for every element $(a,b) \in X_{i,j}$ and place $p_{i,j}^{a,b}$ at the position $x\left(p_{i,j}^{a,b}\right) = 2\eta(i - 1) + a$ and $y\left(p_{i,j}^{a,b}\right) = 2\eta(j - 1) + b$.
Note that the points for every $X_{i,j}$ are \NewText{placed} within a $\eta \times \eta$ square and at a distance of at least $\eta$ from every point for a tile $X_{i',j'}$ with $i' \neq i$ or $j' \neq j$; see \Cref{fig:grid-tiling-example}.
Next, we introduce a vertex $v_{i,j}$ for every $i,j \in [\kappa]$ to $G$ and set $L(v_{i,j}) = \{p_{i,j}^{a,b}\mid (a,b) \in X_{i,j}\}$.
Furthermore, we introduce the edges $\{v_{i,j}, v_{i + 1, j}\}$ and $\{v_{i,j}, v_{i + 1, j}\}$, $i,j \in [\kappa - 1]$.
Observe that $G$ contains a $\kappa \times \kappa$-grid $K$ as subgraph. 
Our reduction uses the following equivalence between solutions:
We select a tuple $(a,b) \in X_{i,j}$ if and only if we place the vertex $v_{i,j}$ onto the point $p_{i,j}^{a,a}$, i.e., $\Drawing(v_{i,j}) = p_{i,j}^{a,b}$.
To guarantee that we obtain a valid solution, we now describe how we ensure that the tuples agree on the first and second value, respectively

\subparagraph*{Encoding the Value Constraints.}
Let $X_{i,j}, X_{i,j+1} \in \mathcal{X}$, $j \in [\kappa - 1]$, be two tiles.
We have to ensure that if $\Drawing(v_{i,j}) = p_{i,j}^{a,b}$ and $\Drawing(v_{i + 1,j}) = p_{i+1,j}^{a',b'}$ then $b = b'$.
To this end, we make use of the $\eta$-wide free space between the points for $A_{i,j}$ and $A_{i,j + 1}$.
Let $\varepsilon > 0$ be a very small constant; we will fix the precise value of $\varepsilon$ towards the end of the construction.
For every $i \in [\kappa - 1]$ and $j \in [\kappa]$, we place $4 (\eta + 1)$ dummy points between the points for the tiles~$X_{i,j}$ and $X_{i + 1, j}$, $2(\eta + 1)$.
More concretely, for every $q \in [\eta]$ we place two of them at the coordinates $(2\eta(i - 1) + \eta + 1, \eta(j - 1) + q - \varepsilon)$ and $(2\eta(i - 1) + \eta + 1, \eta(j - 1) + q + \varepsilon)$, and two at the coordinates $(2\eta i - 1, \eta(j - 1) + q - \varepsilon)$ and $(2\eta i - 1, \eta(j - 1) + q + \varepsilon)$.
We add two further points at $(2\eta(i - 1) + \eta + 1, \eta(j - 1))$ and $(2\eta(i - 1) + \eta + 1, \eta(j - 1) + + \eta + 0.5)$ and at $(2\eta i - 1, \eta(j - 1))$ and $(2\eta i - 1, \eta(j - 1) + \eta + 0.5)$.
Observe that these points leave only $\eta$ corridors between the points for the tiles, each of height $2\varepsilon$.
For each of the above-introduced dummy point $p_d$, we also introduce a corresponding dummy vertex $v_d$, set $L(v_d) = \{p_d\}$, and connect always four dummy vertices to a $C_4$ such that they form a rectangular \emph{obstacle}; see also \Cref{fig:grid-tiling-example}.
We perform symmetrically for every $i \in [\kappa]$ and $j \in [\kappa - 1]$.

Note that the placement of the above-introduced obstacles is fixed as their defining vertices have a single admissible point each.
As indicated above, these obstacles create $2\varepsilon$-tall corridors between points for $X_{i,j}$ and $X_{i + 1,j}$ on the same $y$-coordinate and the edge $\{v_{i,j}, v_{i + 1, j}\}$ has to traverse the space between the points for $X_{i,j}$ and $X_{i,j + 1}$.
However, it can only do so crossing-free by traveling to one of the $2\varepsilon$-tall corridors, for a sufficiently small~$\varepsilon$, effectively ensuring that the points on which the two vertices are embedded on have the same $y$-coordinate.
We now determine the value of $\varepsilon$.
To this end, for some $i \in [\kappa - 1]$, $j \in [\kappa]$ we need to prevent $\Drawing(v_{i,j}) = p_{i,j}^{a,b}$ and $\Drawing(v_{i,j}) = p_{i,j}^{a',b'}$ with $b \neq b'$.
We do that by making sure that no matter which (incorrect) values for $a,b,a',b'$ we pick, we can never enter the corridor for $b$.
Observe that the worst case, i.e., with the lowest $y$-value at the start of the obstacle, we have with $\Drawing(v_{i,j}) = p_{i,j}^{1,b}$ and $\Drawing(v_{i,j}) = p_{i,j}^{\eta,b\pm1}$ and we assume without loss of generality $b' = b + 1$.
The edge between $v_{i,j}$ and $v_{i + 1, j}$ has a $y$-value of $y = b + \frac{b' - b}{a' - a} = b + \frac{1}{2\eta}$ at the start of the obstacle.
Thus, any $\varepsilon$ with $\varepsilon < \frac{1}{2\eta}$, e.g., $\varepsilon = \frac{1}{\eta^2}$, is sufficient for our purpose.
This completes the construction of the instance \InstanceLong of \PartialListPointEmbedding\ and we visualize in \Cref{fig:grid-tiling-example} an example construction.
\begin{restatable}\restateref{lem:extension-w1}{lemma}{lemmaExtensionWOne}
	\label{lem:extension-w1}
	\InstanceLong\ admits a realization if and only $(n, k, \mathcal{A})$ is a yes-instance of \probname{Grid Tiling}.
\end{restatable}
\begin{prooflater}{plemmaExtensionWOne}
	We show both directions separately.

    \proofsubparagraph{($\boldsymbol{\Rightarrow}$)}
    Let $\Drawing$ be a solution to \InstanceLong.
    We now construct a solution to the \probname{Grid Tiling} instance $(\eta, \kappa, \mathcal{X})$ by considering the assignment of the $\kappa\times \kappa$-grid $K$ to $S$ in $\Drawing$.
    For every $i,j \in [\kappa]$, consider the point $p \in S$ with $\Drawing(v_{i,j}) = p$.
    As $L(v_{i,j})$ contains exactly one point for every $(a,b) \in X_{i,j}$ and since $\Drawing$ is a realization, we know that $p = p_{i,j}^{a,b}$ for some $(a,b) \in A_{i,j}$ must hold.
    We now place $a_{i,j}=(a,b) \in X_{i,j}$ in our solution.
    By construction, we select~$\kappa^2$ tuples $a_{i,j}$, exactly one from each tile $X_{i,j} \in \mathcal{X}$.
    It remains to show that if $x_{i,j} = (a,b)$, then $x_{i+1,j} = (a', b)$ and $x_{i,j+1} = (a,b')$ holds for all $i, j \in [\kappa]$ (if the respective tiles exist).
    Towards a contradiction, assume that $x_{i,j} = (a,b)$ and $x_{i+1,j} = (a',b')$ with $b \neq b'$ for some $i \in [\kappa - 1]$ and $j \in [\kappa]$.
    By construction of our solution, this means $\Drawing(v_{i,j}) = p_{i,j}^{a, b}$ and $\Drawing(v_{i + 1,j}) = p_{i + 1,j}^{a', b'}$ with $b \neq b'$.
    In particular, the edge $\{v_{i,j},v_{i + 1, j}\}$ is not horizontal in $\Drawing$.
    As the incident vertices represent tuples from the two tiles $X_{i,j}$ and $X_{i,j + 1}$, the edge $\{v_{i,j},v_{i + 1, j}\}$ must traverse the obstacles between the points $p_{i,j}^{a,b}$ and $p_{i + 1,j}^{a',b'}$.
    However, it can only do so via one of the $2\varepsilon$-tall corridors.
    Since the edge is not horizontal in $\Drawing$, it must cross at least one obstacle by our choice of $\varepsilon$, contradicting the fact that $\Drawing$ is a planar, i.e., it cannot be a solution.
    Hence, $b = b'$ must hold and by symmetric arguments we can also show that for $x_{i, j + 1} = (a'', b'')$ we have $a = a''$.
    Combining all, we conclude that our selection forms a solution to $(\eta, \kappa, \mathcal{X})$.
    
    \proofsubparagraph{($\boldsymbol{\Leftarrow}$)}
    Let $\mathcal{B} = \{x_{i,j} \mid i,j \in [\kappa]\}$ be a solution to $(\eta, \kappa, \mathcal{X})$.
    We construct a solution $\Drawing$ of \InstanceLong as follows.
    First, we use for every obstacle its unique vertex placement.
    Recall that every vertex of an obstacle is single-labeled.
    Afterwards, for every vertex $v_{i,j} \in V(G)$, we set $\Drawing(v_{i,j}) = p_{i,j}^{a,b}$ where $x_{i,j}= \{(a,b)\} \in \mathcal{B}$, for all $i,j \in [\kappa]$.
    As $\mathcal{B}$ is a solution, this is well-defined.
    By construction of \InstanceLong, every vertex is embedded onto an admissible point.
    To see that $\Drawing$ is also planar, we recall, on the one hand, that the obstacles are planar by construction.
    On the other hand, since $\mathcal{B}$ is a solution, we have that if $x_{i,j} = (a,b)$, $x_{i+1,j} = (x',y')$ and $a_{i,j + 1} = (x'',y'')$ then $b = b'$ and $a = a''$ (if the respective tuples exist in $\mathcal{B}$).
    This implies that the edges of the $\kappa\times \kappa$-grid $K$ are axis-aligned in $\Drawing$.
    Hence, they run through the $2\varepsilon$-wide corridors between the obstacles, i.e., no edge of $K$ crosses an obstacle.
    Finally, since $K$ is planar by construction, we conclude that $\Drawing$ is a solution.
\end{prooflater}

Since the construction can be carried out in polynomial time, and only the $\kappa^2$ vertices $v_{i,j}$, $i,j \in [\kappa]$ have more than one admissible point, i.e., are missing, we obtain \Cref{thm:extension-w1}.%

\end{statelater}
\begin{restatable}\restateref{thm:extension-w1}{theorem}{theoremExtensionWOne}
	\label{thm:extension-w1}
	\PartialListPointEmbedding is \W\textup{[1]}-hard parameterized by $\nadd$.
\end{restatable}

\shortLong{
\subparagraph{Adding the Surplus for \textsf{FPT}.}
\Cref{thm:extension-w1} excludes fixed parameter tractability of \PartialListPointEmbedding when parameterized by \nadd alone.
We now extend our parameterization by the \emph{surplus} $\zeta \coloneqq s - n$ of an instance \Instance; note that \PartialListPointEmbedding remains \NP-hard for $\zeta = 0$~\cite{Cab.Pev.2006}.

Our proof strategy is similar to \Cref{thm:extension-xp}:
The assignment for all but the $\nadd$ missing vertices is fixed.
For the remaining $\nadd$ vertices, we branch to determine their assignment to the $\zeta + \nadd$ still unassigned points.
In each resulting branch, every vertex is placed on one of its admissible points and it remains to check whether it corresponds to a solution.
}{
\subsection{Additionally Parameterizing by the Surplus}
\label{sec:extension-fpt}
In the light of \Cref{thm:extension-w1}, which excludes fixed parameter tractability of \PartialListPointEmbedding when parameterized by \nadd alone, we now extend our parameterization by the \emph{surplus} $\zeta \coloneqq s - n$ of an instance \Instance.
Note that \PartialListPointEmbedding remains \NP-hard for $\zeta = 0$; see, for example, the hardness reduction by Cabello~\cite{Cab.Pev.2006}.
In this section, we show that \PartialListPointEmbedding is \FPT\ in $\nadd + \zeta$ using the following proof strategy.
The assignment for all but the $\nadd$ missing vertices is fix.
For the remaining $\nadd$ missing vertices, we branch to determine their assignment to the $\zeta + \nadd$ still unassigned points.
}

\begin{restatable}\restateref{thm:extension-fpt}{theorem}{theoremExtensionFPT}
	\label{thm:extension-fpt}
	\PartialListPointEmbedding can be solved constructively in time $\BigO{(\zeta + {\nadd})^{\nadd} \cdot (s + n \log n)}$ and is thereby \FPT\ %
    in $\nadd + \zeta$. %
\end{restatable}
\begin{prooflater}{ptheoremExtensionFPT}
	We first assign all $n - \nadd$ vertices $v \in V(G) \setminus \Vadd$ to their unique point $p \in S$ with $L(v) = \{p\}$ in $\BigO{n}$ time.
	If this assigns two vertices to the same point, we reject.
    Afterwards, we branch to determine the embedding of the $\nadd$ missing vertices $\Vadd$.
    Observe that $s - (n - \nadd) = s - n + \nadd = \zeta + \nadd$ points remain unassigned.
	Therefore, there are $\BigO{(\zeta + {\nadd})^{\nadd}}$ different branches.
	Each of the branches corresponds to a drawing $\Drawing$.
	Recall that we can check in $\BigO{s + n \log n}$ time whether it is as a realization, i.e., a solution (using, e.g., the Bentley–Ottmann segment intersection test~\cite{dBCvKO.CGA.2008}).
	Since the branching is exhaustive, the statement follows.
\end{prooflater}

\section{Concluding Remarks and Open Problems}
\label{sec:conclusion}
Our results establish a surprisingly tight boundary between tractability and hardness for \ListPointEmbedding.
Still, we see several interesting directions for future research:
First, \Cref{thm:convexDP} requires the graph to be connected, whereas \Cref{thm:hardness-matching} has many connected components.
Studying the complexity of \ConvexListPointEmbedding for graphs with given combinatorial embedding  parameterized by the number of connected components is thus an interesting direction for future work.
Second, \Cref{thm:hardness-paths} establishes \NP-hardness for bi-labeled cycles, which are 2-connected graphs, for general $S$.
We suspect that our hardness can be adapted for $3$-connected graphs, but leave this open for future work.
\NewText{Third, our hardness reduction behind \Cref{thm:hardness-paths} requires more points than vertices and when introducing dummy vertices, we have to ensure that the graph remains a path or a cycle.
We consider the problem of adapting the hardness construction such that the path (or cycle) has $s$ vertices to be a further direction for future work.
}

\bibliography{references}

@Article{Cab.Pev.2006,
  author    = {Cabello, Sergio},
  journal   = {Journal of Graph Algorithms and Applications},
  title     = {Planar embeddability of the vertices of a graph using a fixed point set is NP-hard},
  year      = {2006},
  issn      = {1526-1719},
  number    = {2},
  pages     = {353--363},
  volume    = {10},
  doi       = {10.7155/jgaa.00132},
  publisher = {Journal of Graph Algorithms and Applications},
}

@InProceedings{BV.pse.2012,
  author     = {Biedl, Therese and Vatshelle, Martin},
  booktitle  = {Proceedings of the twenty-eighth annual symposium on Computational geometry},
  title      = {The point-set embeddability problem for plane graphs},
  year       = {2012},
  month      = jun,
  pages      = {41--50},
  publisher  = {ACM},
  series     = {SoCG ’12},
  collection = {SoCG ’12},
  doi        = {10.1145/2261250.2261257},
}

@InProceedings{DM.HPS.2012,
  author    = {Durocher, Stephane and Mondal, Debajyoti},
  booktitle = {WALCOM: Algorithms and Computation},
  title     = {On the Hardness of Point-Set Embeddability},
  year      = {2012},
  pages     = {148--159},
  publisher = {Springer Berlin Heidelberg},
  doi       = {10.1007/978-3-642-28076-4_16},
  isbn      = {9783642280764},
  issn      = {1611-3349},
}

@InProceedings{DGDL+.kCP.2007,
  author    = {Di Giacomo, Emilio and Didimo, Walter and Liotta, Giuseppe and Meijer, Henk and Trotta, Francesco and Wismath, Stephen K.},
  booktitle = {Graph Drawing},
  title     = {k-Colored Point-Set Embeddability of Outerplanar Graphs},
  year      = {2007},
  pages     = {318--329},
  publisher = {Springer Berlin Heidelberg},
  doi       = {10.1007/978-3-540-70904-6_31},
  isbn      = {9783540709046},
  issn      = {1611-3349},
}

@Article{DGLT.DCG.2008,
  author    = {Di Giacomo, Emilio and Liotta, Giuseppe and Trotta, Francesco},
  journal   = {Algorithmica},
  title     = {Drawing Colored Graphs with Constrained Vertex Positions and Few Bends per Edge},
  year      = {2008},
  issn      = {1432-0541},
  month     = dec,
  number    = {4},
  pages     = {796--818},
  volume    = {57},
  doi       = {10.1007/s00453-008-9255-2},
  publisher = {Springer Science and Business Media LLC},
}

@InProceedings{FGL+.PSE.2013,
  author    = {Frati, Fabrizio and Glisse, Marc and Lenhart, William J. and Liotta, Giuseppe and Mchedlidze, Tamara and Nishat, Rahnuma Islam},
  booktitle = {Graph Drawing},
  title     = {Point-Set Embeddability of 2-Colored Trees},
  year      = {2013},
  pages     = {291--302},
  publisher = {Springer Berlin Heidelberg},
  doi       = {10.1007/978-3-642-36763-2_26},
  isbn      = {9783642367632},
  issn      = {1611-3349},
}

@Article{BDGL.Dcg.2008,
  author    = {Badent, Melanie and Di Giacomo, Emilio and Liotta, Giuseppe},
  journal   = {Theoretical Computer Science},
  title     = {Drawing colored graphs on colored points},
  year      = {2008},
  issn      = {0304-3975},
  month     = nov,
  number    = {2–3},
  pages     = {129--142},
  volume    = {408},
  doi       = {10.1016/j.tcs.2008.08.004},
  publisher = {Elsevier BV},
}

@Article{ADBF+.TPP.2015,
  author    = {Angelini, Patrizio and Di Battista, Giuseppe and Frati, Fabrizio and Jelínek, Vít and Kratochvíl, Jan and Patrignani, Maurizio and Rutter, Ignaz},
  journal   = {ACM Transactions on Algorithms},
  title     = {Testing Planarity of Partially Embedded Graphs},
  year      = {2015},
  issn      = {1549-6333},
  month     = apr,
  number    = {4},
  pages     = {1--42},
  volume    = {11},
  doi       = {10.1145/2629341},
  publisher = {Association for Computing Machinery (ACM)},
}

@InBook{Pat.EPS.2006,
  author    = {Patrignani, Maurizio},
  pages     = {380--385},
  publisher = {Springer Berlin Heidelberg},
  title     = {On Extending a Partial Straight-Line Drawing},
  year      = {2006},
  isbn      = {9783540316671},
  booktitle = {Graph Drawing},
  doi       = {10.1007/11618058_34},
  issn      = {1611-3349},
}

@Article{MNR.ECP.2015,
  author    = {Mchedlidze, Tamara and Nöllenburg, Martin and Rutter, Ignaz},
  journal   = {Algorithmica},
  title     = {Extending Convex Partial Drawings of Graphs},
  year      = {2015},
  issn      = {1432-0541},
  month     = jun,
  number    = {1},
  pages     = {47--67},
  volume    = {76},
  doi       = {10.1007/s00453-015-0018-6},
  publisher = {Springer Science and Business Media LLC},
}

@Misc{SCM.ETR.2024,
  author    = {Schaefer, Marcus and Cardinal, Jean and Miltzow, Tillmann},
  title     = {The Existential Theory of the Reals as a Complexity Class: A Compendium},
  year      = {2024},
  copyright = {Creative Commons Attribution 4.0 International},
  doi       = {10.48550/ARXIV.2407.18006},
  publisher = {arXiv},
}

@InProceedings{KW.EVP.1999,
  author    = {Kaufmann, Michael and Wiese, Roland},
  booktitle = {Graph Drawing},
  title     = {Embedding Vertices at Points: Few Bends Suffice for Planar Graphs},
  year      = {1999},
  pages     = {165--174},
  publisher = {Springer Berlin Heidelberg},
  doi       = {10.1007/3-540-46648-7_17},
  isbn      = {9783540466482},
  issn      = {1611-3349},
}

@Article{PW.EPG.2001,
  author    = {Pach, János and Wenger, Rephael},
  journal   = {Graphs and Combinatorics},
  title     = {Embedding Planar Graphs at Fixed Vertex Locations},
  year      = {2001},
  issn      = {0911-0119},
  month     = dec,
  number    = {4},
  pages     = {717--728},
  volume    = {17},
  doi       = {10.1007/pl00007258},
  publisher = {Springer Science and Business Media LLC},
}

@Article{DGLT.EGT.2006,
  author       = {Emilio {Di Giacomo} and
                  Giuseppe Liotta and
                  Francesco Trotta},
  title        = {On Embedding a Graph on Two Sets of Points},
  journal      = {Int. J. Found. Comput. Sci.},
  volume       = {17},
  number       = {5},
  pages        = {1071--1094},
  year         = {2006},
  doi          = {10.1142/S0129054106004273},
}

@Book{Die.GT4.2012,
  author    = {Reinhard Diestel},
  publisher = {Springer},
  title     = {{G}raph {T}heory, {4th} {E}dition},
  year      = {2012},
  isbn      = {978-3-642-14278-9},
  series    = {Graduate {T}exts in {M}athematics},
  volume    = {173},
}

@Book{dBCvKO.CGA.2008,
  author    = {Mark de Berg and Otfried Cheong and Marc J. van Kreveld and Mark H. Overmars},
  publisher = {Springer},
  title     = {Computational {G}eometry: {A}lgorithms and {A}pplications, 3rd {E}dition},
  year      = {2008},
}

@Article{CKX.Iub.2010,
  author  = {Jianer Chen and Iyad A. Kanj and Ge Xia},
  journal = {Theoretical Computer Science},
  title   = {Improved upper bounds for vertex cover},
  year    = {2010},
  number  = {40–42},
  pages   = {3736--3756},
  volume  = {411},
  doi     = {10.1016/j.tcs.2010.06.026},
}

@Article{HK.AMM.1973,
  author  = {Hopcroft, John E. and Karp, Richard M.},
  journal = {SIAM Journal on Computing},
  title   = {An $n^{5/2} $ Algorithm for Maximum Matchings in Bipartite Graphs},
  year    = {1973},
  number  = {4},
  pages   = {225--231},
  volume  = {2},
  doi     = {10.1137/0202019},
}

@Book{CFK+.PA.2015,
  author    = {Marek Cygan and Fedor V. Fomin and Lukasz Kowalik and Daniel Lokshtanov and D{\'{a}}niel Marx and Marcin Pilipczuk and Michal Pilipczuk and Saket Saurabh},
  publisher = {Springer},
  title     = {Parameterized {A}lgorithms},
  year      = {2015},
  isbn      = {978-3-319-21274-6},
  doi       = {10.1007/978-3-319-21275-3},
}

@book{vazirani2001approximation,
  title={Approximation algorithms},
  author={Vazirani, Vijay V},
  volume={1},
  year={2001},
  publisher={Springer}
}

@InProceedings{BergKhosravi2010,
author="de Berg, Mark
and Khosravi, Amirali",
editor="Thai, My T.
and Sahni, Sartaj",
title="Optimal Binary Space Partitions in the Plane",
booktitle="Computing and Combinatorics",
year="2010",
publisher="Springer Berlin Heidelberg",
address="Berlin, Heidelberg",
pages="216--225",
}

@article{GMP+.Ept.1991,
  title={Embedding a planar triangulation with vertices at specified points},
  author={Peter Gritzmann and Bojan Mohar and J{\'a}nos Pach and Richard Pollack},
  journal={American Mathematical Monthly},
  year={1991},
  volume={98},
  pages={165-166},
  doi={10.2307/2323956}
}

@inproceedings{CU.SLE.1996,
author = {Casta\~{n}eda, Netzahualcoyotl and Urrutia, Jorge},
title = {Straight Line Embeddings of Planar Graphs on Point Sets},
year = {1996},
publisher = {Carleton University Press},
booktitle = {Proceedings of the 8th Canadian Conference on Computational Geometry},
pages = {312–318},
numpages = {7}
}

@article{ABB+.SUP.2018,
  author       = {Patrizio Angelini and
                  Till Bruckdorfer and
                  Giuseppe Di Battista and
                  Michael Kaufmann and
                  Tamara Mchedlidze and
                  Vincenzo Roselli and
                  Claudio Squarcella},
  title        = {Small Universal Point Sets for k-Outerplanar Graphs},
  journal      = {Discret. Comput. Geom.},
  volume       = {60},
  number       = {2},
  pages        = {430--470},
  year         = {2018},
  doi          = {10.1007/S00454-018-0009-X},
}

@article{Kur.Alb.2004,
  author       = {Maciej Kurowski},
  title        = {A 1.235 lower bound on the number of points needed to draw all \emph{n}-vertex
                  planar graphs},
  journal      = {Inf. Process. Lett.},
  volume       = {92},
  number       = {2},
  pages        = {95--98},
  year         = {2004},
  doi          = {10.1016/J.IPL.2004.06.009},
}

@article{BCD+.SUP.2014,
  author       = {Michael J. Bannister and
                  Zhanpeng Cheng and
                  William E. Devanny and
                  David Eppstein},
  title        = {Superpatterns and Universal Point Sets},
  journal      = {J. Graph Algorithms Appl.},
  volume       = {18},
  number       = {2},
  pages        = {177--209},
  year         = {2014},
  doi          = {10.7155/JGAA.00318},
}

@article{DPP.Hdp.1990,
  author       = {Hubert de Fraysseix and
                  J{\'{a}}nos Pach and
                  Richard Pollack},
  title        = {How to draw a planar graph on a grid},
  journal      = {Comb.},
  volume       = {10},
  number       = {1},
  pages        = {41--51},
  year         = {1990},
  doi          = {10.1007/BF02122694},
}

@inproceedings{Sch.EPG.1990,
  author       = {Walter Schnyder},
  editor       = {David S. Johnson},
  title        = {Embedding Planar Graphs on the Grid},
  booktitle    = {Proceedings of the First Annual {ACM-SIAM} Symposium on Discrete Algorithms,
                  22-24 January 1990, San Francisco, California, {USA}},
  pages        = {138--148},
  publisher    = {{SIAM}},
  year         = {1990},
}

@article{CHK.UPS.2015,
  author       = {Jean Cardinal and
                  Michael Hoffmann and
                  Vincent Kusters},
  title        = {On Universal Point Sets for Planar Graphs},
  journal      = {J. Graph Algorithms Appl.},
  volume       = {19},
  number       = {1},
  pages        = {529--547},
  year         = {2015},
  doi          = {10.7155/JGAA.00374},
}

@article{SSS.NUP.2020,
  author       = {Manfred Scheucher and
                  Hendrik Schrezenmaier and
                  Raphael Steiner},
  title        = {A Note on Universal Point Sets for Planar Graphs},
  journal      = {J. Graph Algorithms Appl.},
  volume       = {24},
  number       = {3},
  pages        = {247--267},
  year         = {2020},
  doi          = {10.7155/JGAA.00529},
}

@article{GGL+.CCC.2020,
  author       = {Emilio {Di Giacomo} and
                  Leszek Gasieniec and
                  Giuseppe Liotta and
                  Alfredo Navarra},
  title        = {On the curve complexity of 3-colored point-set embeddings},
  journal      = {Theor. Comput. Sci.},
  volume       = {846},
  pages        = {114--140},
  year         = {2020},
  doi          = {10.1016/J.TCS.2020.09.027},
}

@article{GJL.CPS.2021,
  author       = {Emilio {Di Giacomo} and
                  Jaroslav Hancl Jr. and
                  Giuseppe Liotta},
  title        = {2-colored point-set embeddings of partial 2-trees},
  journal      = {Theor. Comput. Sci.},
  volume       = {896},
  pages        = {31--45},
  year         = {2021},
  doi          = {10.1016/J.TCS.2021.09.045},
}

@article{CET+.DGP.2012,
  author       = {Erin W. Chambers and
                  David Eppstein and
                  Michael T. Goodrich and
                  Maarten L{\"{o}}ffler},
  title        = {Drawing Graphs in the Plane with a Prescribed Outer Face and Polynomial
                  Area},
  journal      = {J. Graph Algorithms Appl.},
  volume       = {16},
  number       = {2},
  pages        = {243--259},
  year         = {2012},
  doi          = {10.7155/JGAA.00257},
}

@inproceedings{eghkn-ep1d-20,
	Author = {Eiben, Eduard and Ganian, Robert and Hamm, Thekla and Klute, Fabian and Nöllenburg, Martin},
	Title = {Extending Partial 1-Planar Drawings},
	Booktitle = {Automata, Languages, and Programming (ICALP'20)},
	Year = {2020},
    Editor = {Czumaj, Artur and Dawar, Anuj and Merelli, Emanuela},
    Volume = {168},
    Pages = {43:1--43:19},
    Series = {LIPIcs},
    Publisher = {Schloss Dagstuhl--Leibniz-Zentrum für Informatik},
    Doi = {10.4230/LIPIcs.ICALP.2020.43},
    Ee = {2004.12222},
}

@inproceedings{dfgn-pcesl-24,
	Author = {Depian, Thomas and Fink, Simon D. and Ganian, Robert and Nöllenburg, Martin},
	Title = {The Parameterized Complexity of Extending Stack Layouts},
	Booktitle = {Graph Drawing and Network Visualization (GD'24)},
	Year = {2024},
    Editor = {Felsner, Stefan and Klein, Karsten},
    Volume = {320},
    Pages = {12:1--12:17},
    Series = {LIPIcs},
    Publisher = {Schloss Dagstuhl -- Leibniz-Zentrum für Informatik},
    _Url = {https://arxiv.org/abs/2409.02833},
    Doi = {10.4230/LIPIcs.GD.2024.12},
    Ee = {2409.02833},
}

@Article{HT.EAG.1973,
  author  = {Hopcroft, John and Tarjan, Robert},
  journal = {Communications of the ACM},
  title   = {Efficient {A}lgorithms for {G}raph {M}anipulation {[H]} ({A}lgorithm 447)},
  year    = {1973},
  number  = {6},
  pages   = {372--378},
  volume  = {16},
  doi     = {10.1145/362248.362272},
}

@article{Mit.LAR.1979,
  author       = {Sandra L. Mitchell},
  title        = {Linear Algorithms to Recognize Outerplanar and Maximal Outerplanar Graphs},
  journal      = {Information Processing Letters},
  volume       = {9},
  number       = {5},
  pages        = {229--232},
  year         = {1979},
  doi          = {10.1016/0020-0190(79)90075-9},
}

@Article{DFGN.PCE.2026,
  author  = {Depian, Thomas and Fink, Simon D. and Ganian, Robert and Nöllenburg, Martin},
  journal = {Journal of Graph Algorithms and Applications},
  title   = {The Parameterized Complexity Of Extending Stack Layouts},
  year    = {2026},
  number  = {3},
  pages   = {39--78},
  volume  = {29},
  doi     = {10.7155/jgaa.v29i3.3221},
}

@inproceedings{bgkmn-eopgdft-23,
	Author = {Bhore, Sujoy and Ganian, Robert and Khazaliya, Liana and Montecchiani, Fabrizio and Nöllenburg, Martin},
	Title = {Extending Orthogonal Planar Graph Drawings is Fixed-Parameter Tractable},
	Booktitle = {Computational Geometry (SoCG'23)},
	Year = {2023},
    Editor = {Chambers, Erin W. and Gudmundsson, Joachim},
    Volume = {258},
    Pages = {18:1--18:16},
    Series = {LIPIcs},
    Publisher = {Schloss Dagstuhl -- Leibniz-Zentrum für Informatik},
    Doi = {10.4230/LIPIcs.SoCG.2023.18},
    Ee = {2302.10046},
}

@article{bgkmn-eopgdft-25,
	Author = {Bhore, Sujoy and Ganian, Robert and Khazaliya, Liana and Montecchiani, Fabrizio and Nöllenburg, Martin},
	Title = {Extending Orthogonal Planar Graph Drawings is Fixed-parameter Tractable},
	Journal = {J. Computational Geometry},
	Year = {2024},
    Volume = {15},
    Number = {2},
    Pages = {3--39},
    Doi = {10.20382/jocg.v15i2a1},
}

@inproceedings{PLP,
	Author = {Depian, Thomas and Fink, Simon D. and Klemz, Boris and Ganian, Robert and Nöllenburg, Martin and Sieper, Marie Diana},
	Title = {Partial level planarity parameterized by the size of the missing graph},
	Booktitle = {Proc. 41st European Workshop on Computational Geometry (EuroCG'25)},
	Year = {2025},
    Editor = {Balko, Martin and Kratochvíl, Jan and Liotta, Giuseppe},
    Pages = {50:1--50:10},
    url = {https://kam.mff.cuni.cz/conferences/eurocg2025/booklet2025.pdf}
}

@InProceedings{QLE,
  author    = {Thomas Depian and Simon D. Fink and Robert Ganian and Martin N{\"{o}}llenburg},
  booktitle = {Proc. 51st International Workshop on Graph-Theoretic Concepts in Computer Science (WG'25)},
  title     = {The {P}eculiarities of {E}xtending {Q}ueue {L}ayouts},
  year      = {2026},
  doi       = {10.1007/978-3-032-11835-6_13},
  pages     = {177--191},
  editor    = {Henning Fernau and Philipp Kindermann},
}

\end{document}